\documentclass[11pt]{article}

\usepackage[]{inputenc}
\usepackage[T1]{fontenc}
\usepackage{hyperref}
\usepackage{cite}

\usepackage{color}

\usepackage{fullpage}

\usepackage{amsmath}
\usepackage{amssymb}
\usepackage[toc,page]{appendix}

\usepackage{amsthm}
\usepackage{float}
\usepackage{wrapfig}
\usepackage{caption}
\usepackage{graphicx}
\usepackage{pgfplots}
\usepackage{subcaption}
\usepackage{tikz,tikz-3dplot}
\tdplotsetmaincoords{80}{45}
\tdplotsetrotatedcoords{-90}{180}{-90}

\tikzset{surface1/.style={draw=blue!70!black, fill=blue!40!white, fill opacity=.6}}
\tikzset{surface2/.style={draw=red!70!black, fill=red!40!white, fill opacity=.6}}
\tikzset{surface3/.style={draw=green!70!black, fill=green!40!white, fill opacity=.6}}

\usetikzlibrary{intersections,fadings,decorations.pathreplacing}

  \pgfplotsset{
        compat=1.8}

\newtheorem{theorem}{Theorem}[section]
\newtheorem{lemma}{Lemma}[section]

\theoremstyle{definition}
\newtheorem*{definition}{Definition}

\def \be {\begin{equation}}
\def \ee {\end{equation}}
\def \bea {\begin{eqnarray}}
\def \eea {\end{eqnarray}}

\title{On the construction of asymptotically flat initial data in scalar-tensor effective field theory}

\author{{\'A}ron D. Kov{\'a}cs\\{\small Department of Applied Mathematics and Theoretical Physics, University of Cambridge}\\ {\small Wilberforce Road, Cambridge CB3 0WA, United Kingdom} \\ adk42@cam.ac.uk}

\begin{document}
\maketitle

\begin{abstract}
We study the constraint equations for a class of scalar-tensor effective field theories of gravity, including the operators up to $4$ derivatives in the action ($4\partial$ST). We extend the conformal transverse traceless and conformal thin sandwich methods of General Relativity to rewrite the constraint equations of the scalar-tensor theory as a set of elliptic partial differential equations. It is shown that, at weak coupling, a unique solution exists to the corresponding elliptic boundary value problems on asymptotically Euclidean initial slices under similar conditions as in the case of General Relativity. Furthermore, we find a generalization of the Bowen-York solution for $4\partial$ST theories, too. These results demonstrate that standard methods for constructing initial data in General Relativity are applicable (with minimal modification) to weakly coupled $4\partial$ST theories.
\end{abstract}

\section{Introduction}

The direct observation of gravitational waves (GW) produced by compact binary systems \cite{Abbott:2016blz} has inspired a tremendous amount of work on gravitational physics. In particular, many are excited by the possibility to use GW observations to test General Relativity (GR) or (one day perhaps) to find convincing evidence for new physics beyond GR. In order to do this, however, we need to know how the gravitational wave signals predicted by alternative theories differ from those of GR. 

To test General Relativity, one may turn to the Effective Field Theory (EFT) philosophy of model-building\cite{Burgess:2003jk,Endlich:2017tqa}. According to the bottom-up EFT approach, we can parameterize our ignorance about fundamental physics by enumerating all the operators in the effective action with the desired symmetries and field content. In the effective action, operators containing higher derivatives are suppressed by powers of a strong coupling energy scale. To study low energy processes (i.e. at energies well below the strong coupling scale) it is sufficient to keep only the first few terms in the effective action. Therefore, the classical equations of motion of the weakly coupled, truncated effective theory are expected to describe low energy physics accurately. Although the corrections to the equations of motion arising from the higher derivative terms are small at weak coupling, these corrections may be important in certain situations. Small effects may accumulate over time, producing large observable deviations from the leading order theory \cite{Flanagan:1996gw}.

To make predictions using an EFT of gravity, the theory must satisfy some minimal mathematical requirements. For example, the theory must possess a well-posed initial value formulation. This means that given a suitable class of initial data, the equations of motion of the theory must be able to predict the future evolution of the data. More precisely, there must exist a unique solution to the (gauge-fixed) equations of motion with the specified initial data so that the solution depends continuously on the data (in a suitable norm). This is not only an important mathematical consistency condition but also necessary to perform numerical simulations.

Consider first the EFT of vacuum gravity. The leading order correction to the ($2$-derivative) Einstein-Hilbert action contains six derivatives in $d=4$ spacetime dimensions (after field redefinitions) \cite{Endlich:2017tqa}. One difficulty with this EFT is that it has higher than second order equation of motion. It is currently an open question how to obtain a well-posed initial value formulation for such a theory, hindering the possibility to explore its dynamics. (See, however, \cite{Cayuso:2020lca} for some potential remedies of this issue in a spherically symmetric setting.)

It appears to be easier to make progress when gravity is coupled to a real scalar field. In this case, the leading order term (containing $2$ derivatives) in the action is the Einstein-Hilbert term with the standard scalar kinetic term. The leading order corrections that respect parity symmetry appear at $4$ derivatives after field redefinitions \cite{Weinberg:2008hq}. Including these $4$-derivative corrections, the effective action can be brought to the form 
\be
\label{4dST}
 S = \frac{1}{16\pi G} \int d^4 x \sqrt{-g} \left( - V(\phi)+ R + X +  \epsilon_1 f_1(\phi) X^2 +  \epsilon_2 f_2(\phi) {\cal L}_{\rm GB}  \right)
\ee
where $X \equiv -(1/2) g^{\mu\nu} \partial_\mu \phi \partial_\nu \phi$ is the kinetic term of the minimally coupled scalar, $V(\phi)$, $f_1(\phi)$ and $f_2(\phi)$ are arbitrary smooth functions, $\epsilon_1$ and $\epsilon_2$ are dimensionless constant parameters introduced for later convenience and
\be
\label{LGB}
 {\cal L}_{GB}=\frac14 \delta^{\mu_1\mu_2\mu_3\mu_4}_{\nu_1\nu_2\nu_3\nu_4}R_{\mu_1\mu_2}{}^{\nu_1\nu_2}R_{\mu_3\mu_4}{}^{\nu_3\nu_4}.
\ee
is the Gauss-Bonnet invariant. Here the generalised Kronecker delta is given by
\be
\delta^{\mu_1...\mu_n}_{\nu_1...\nu_n}\equiv n! \delta^{\mu_1}_{[\nu_1}\ldots \delta^{\mu_n}_{\nu_n]}.
\ee
Henceforth we shall refer to the theory \eqref{4dST} as $4$-derivative scalar-tensor or $4\partial$ST theory. The theory has second order equations of motion
\bea
E^{\mu}{}_\nu &= & G^{\mu}{}_\nu-\frac12 \left(X-V(\phi)+\epsilon_1 f_1(\phi)X^2\right)\delta^\mu_\nu-\frac12(1+2\epsilon_1 f_1(\phi)X)\nabla^\mu\phi\nabla_\nu\phi \nonumber \\
& &+\left(\epsilon_2 f_2''(\phi)\nabla_{\nu_1}\phi\nabla^{\mu_1}\phi+\epsilon_2 f_2'(\phi)\nabla_{\nu_1}\nabla^{\mu_1}\phi\right)\delta_{\nu\nu_1\nu_2\nu_3}^{\mu\mu_1\mu_2\mu_3}R_{\mu_2\mu_3}{}^{\nu_2\nu_3} \label{4dst_grav_eom}
\eea
\bea
E_\phi &=&-(1+6\epsilon_1 f_1(\phi)X)\Box_g\phi+V'(\phi)-2\epsilon_1 f_1(\phi)\delta_{\nu_1\nu_2}^{\mu_1\mu_2}\nabla_{\mu_1}\phi\nabla^{\nu_1}\phi\nabla_{\mu_2}\nabla^{\nu_2}\phi \nonumber \\
& &+3\epsilon_1 f_1^\prime(\phi)X^2-\frac14 \epsilon_2 f_2'(\phi) \delta_{\nu_1\nu_2\nu_3\nu_4}^{\mu_1\mu_2\mu_3\mu_4} R_{\mu_1\mu_2}{}^{\nu_1\nu_2} R_{\mu_3\mu_4}{}^{\nu_3\nu_4} \label{4dst_scalar_eom}
\eea
i.e. it is of the Horndeski class \cite{Horndeski1974,Kobayashi:2011nu}. It was recently shown \cite{Kovacs:2020pns,Kovacs:2020ywu} that weakly coupled $4\partial$ST theories (or more generally, weakly coupled Horndeski \cite{Horndeski1974} theories) admit a locally well-posed initial value formulation, hence satisfying the important mathematical requirement of predictability mentioned above.

The scalar-tensor theory \eqref{4dST} is also interesting from a physics point of view because it admits black hole solutions with scalar hair for several choices of the coupling function $f_2(\phi)$ (see e.g. \cite{Kanti:1995vq,Kleihaus:2011tg,Kleihaus:2015aje,Sotiriou:2014pfa,Delgado:2020rev,Doneva:2017bvd,Cunha:2019dwb}). This means that there are black hole solutions of some of these theories that differ from the black hole solutions of GR which may have important observational implications. Hence, there is plenty of motivation to study the $4\partial$ST theories in more detail.

To explore the dynamics of $4\partial$ST theories, it is important to consider another requirement: the possibility to construct initial data that gives a good approximation of a realistic astrophysical system when evolved in time. Similarly to the initial value problem, this is a mathematical consistency problem that is also relevant for numerical simulations. In General Relativity, the initial data is subject to constraint equations obtained as follows. Let $\Sigma$ be a spacelike hypersurface in a globally hyperbolic spacetime with (future-directed) unit normal $n^\mu$ and let $\gamma$ be the induced metric on $\Sigma$. Then the projections
\be
{\cal H}^{\rm GR}\equiv G_{\mu\nu}n^\mu n^\nu=0 \qquad \qquad {\rm and} \qquad \qquad {\cal M}^{\rm GR}_\mu\equiv G_{\alpha\beta}n^\alpha \gamma^\beta_\mu=0
\ee
of Einstein's equations, called the Hamiltonian and momentum constraints, contain no second time derivatives of the metric. Similarly, in a $4\partial$ST theory, the projections
\be\label{constraint_proj}
{\cal H}\equiv E_{\mu\nu}n^\mu n^\nu=0 \qquad \qquad {\rm and} \qquad \qquad {\cal M}_\mu\equiv E_{\alpha\beta}n^\alpha \gamma^\beta_\mu=0
\ee
of the $4\partial$ST equation of motion \eqref{4dst_grav_eom} are constraint equations: they contain no second time derivatives of the metric and the scalar field.

It should be noted that it is possible to perform numerical simulations in a $4\partial$ST theory without a detailed knowledge of the properties of the $4\partial$ST constraint equations. The dynamics of binary black hole systems in a shift symmetric $4\partial$ST theory was studied numerically in a recent work of East and Ripley \cite{East:2020hgw}. They focused on the theory with $V(\phi)=f_1(\phi)=0$ and $f_2(\phi)=\lambda \phi$ ($\lambda$ is a constant) which admits black hole solutions with scalar hair. The initial data used in \cite{East:2020hgw} is based on the observation that the $4\partial$ST constraint equations \eqref{constraint_proj} reduce to the Hamiltonian and momentum constraints of vacuum General Relativity provided that one chooses the initial values of the scalar field and its time derivative to be zero on the entire initial data surface. Hence one can use a general relativistic initial data for a binary black hole system (with a trivial initial scalar field configuration) and evolve this data in time. With the choice of couplings mentioned above, the black holes scalarize dynamically during the evolution. Nevertheless, one may be interested in a more general class of solutions to the constraint equations. For example, one may wish to construct initial data that is closer to a binary system of scalarized black holes.

In General Relativity there is a vast literature on the initial data problem (see e.g. \cite{Choquet-Bruhat:2009xil,Cook:2000vr,num_rel,Tichy:2016vmv} and references therein). Most of the successful approaches, such as the conformal transverse traceless (CTT) \cite{York:1973ia} or the conformal thin sandwich (CTS) \cite{York:1998hy} method, are based on a conformal rescaling of certain variables. The main idea is that the constraint equations form a system of elliptic partial differential equations (PDEs) in some of these conformally rescaled variables, whereas other variables play the role of freely specifiable sources ("free data"). Once a useful set of variables is identified, the construction of initial data comes down to the following two steps. Firstly, it must be demonstrated that under the right conditions (on the free data and the topology of the initial slice) the elliptic equations yield a unique solution for the rest of the variables.\footnote{The failure of uniqueness is not necessarily a problem from a mathematical point of view, as long as one can make a clear interpretation of what each solution represents. However, it might cause problems (failure of convergence) in numerical simulations (see e.g. the Summary and discussion section of \cite{Baumgarte:2006ug}).} Secondly, the free data must be chosen such that it corresponds to an astrophysically realistic system.

In this paper, we address these two problems in the context of weakly coupled $4\partial$ST theories. We apply the above mentioned two conformal methods to $4\partial$ST gravity and rewrite the $4\partial$ST constraints in terms of the conformally rescaled variables. Since we are mainly interested in isolated systems such as a binary black hole system, in this paper we shall study these elliptic equations on asymptotically Euclidean initial slices. Although a general $4\partial$ST theory does not have as good conformal properties as GR, the implicit function theorem (in Banach spaces) guarantees existence and uniqueness of solutions to the conformally formulated $4\partial$ST constraints at weak coupling.

A special solution of the conformally formulated constraints of vacuum General Relativity is the Bowen-York initial data \cite{Bowen:1980yu} (and its variations \cite{Brandt:1997tf}) that is often used in numerical relativity studies to construct approximate initial data for multiple boosted and rotating black holes. The significance of the Bowen-York method is that it provides a simple analytical solution to the momentum constraint and reduces the problem to the numerical integration of the Hamiltonian constraint. We show that the Bowen-York method can be generalized to $4\partial$ST theories. In this case, we find that the momentum constraint can be solved exactly for the canonical momenta conjugate to the scalar field $\phi$ and the spatial metric $\gamma_{ij}$. This trick reduces the problem of solving the constraints to solving the Hamiltonian constraint and a system of algebraic equations. The role of the extra algebraic equations is to obtain the extrinsic curvature $K_{ij}=-\tfrac12 {\cal L}_n \gamma_{ij}$ and the quantity $K_\phi=-\tfrac12 {\cal L}_n\phi$ from the canonical momenta. Once again, the implicit function theorem implies that the resulting system of equation can be solved, at least for small couplings. In numerical applications, this system can be solved iteratively in the couplings.

The paper is organised as follows. In Section \ref{Sec:CTT}, we describe the conformal transverse traceless method and extend it to $4\partial$ST theories. In more detail, section \ref{sec:gr_ctt} contains the definition of the conformal variables used throughout this paper and the derivation of the CTT version of the constraint equations in General Relativity. Next, in section \ref{sec:gr_ctt_math}, we review some mathematical definitions and theorems on the well-posedness of the boundary value problem of the CTT system for asymptotically Euclidean initial data slices. Section \ref{sec:gr_by} reviews the Bowen-York initial data for General Relativity which is based on the CTT constraint equations. The rest of Section \ref{Sec:CTT} details the extension of the CTT method to $4\partial$ST theories. After providing the CTT constraint equations for the scalar-tensor theory \eqref{4dST} in section \ref{sec:4dst_ctt}, we state and prove a theorem on the well-posedness of the corresponding boundary value problem. Although it is aimed to be self-contained and accessible to readers outside of the mathematical relativity community, section \ref{sec:4dst_ctt} is not essential for understanding the method of initial data construction. Finally, we conclude the discussion of the CTT approach by presenting an extension of the Bowen-York "puncture" data to $4\partial$ST theories.

Section \ref{Sec:CTS} is concerned with the original conformal thin sandwich method. From a mathematical point of view, this is very similar to the CTT approach. However, the CTS method deserves a separate treatment due to its physical significance. After reviewing the CTS equations that gives the basis of a popular way of initial data construction, we state the corresponding existence and uniqueness theorems on asymptotically Euclidean manifolds both in General Relativity and in $4\partial$ST theories. The section is concluded with some remarks on numerical relativity applications.

Sometimes, an extension of the original conformal thin sandwich method (dubbed as XCTS) is used to find initial data in numerical studies \cite{Pfeiffer:2002iy}. The significance of this method is that in General Relativity it provides an elegant way to construct initial data for a binary black hole spacetime in a corotating coordinate system such that the black holes move along quasicircular orbits. However, the mathematical theory of the extended CTS system is more complicated than in the case of the CTT or the original CTS methods. In particular, it is known that this system fails to have a unique solution for certain choices of the "free data". For this reason, we merely put forward a proposal on how to adapt the extended CTS approach to $4\partial$ST theories and we only briefly discuss the expected properties of the resulting system.

Finally, we mention some possible follow-up work in section \ref{sec:discussion}. Appendices \ref{app:adm} and \ref{app:conf} contain some identities that may be helpful for performing the $3+1$ and conformal decompositions of the $4\partial$ST equations.

{\bf Notations.} We follow the conventions of \cite{Wald:1984rg}, unless stated otherwise. Lower case Greek indices $\mu,\nu,\ldots$ are tensor indices and run from $0$ to $3$, The indices $i,j,\ldots$ run from $1$ to $3$ and stand for indices of tensors projected to the initial data hypersurface $\Sigma$.

\section{The conformal transverse traceless decomposition}\label{Sec:CTT}

\subsection{Description of the method in General Relativity}\label{sec:gr_ctt}

In this section, we review the conformal transverse traceless (CTT) approach introduced by York \cite{York:1973ia} to construct an asymptotically flat initial data set. Let $(M,g)$ be a smooth, $4$-dimensional, globally hyperbolic spacetime. An initial data set is a triple $(\Sigma, \gamma_{ij}, K_{ij})$ where $\Sigma$ is a smooth 3-dimensional spacelike submanifold of $M$, $\gamma_{ij}$ is a Riemannian metric on $\Sigma$ and $K_{ij}$ is the extrinsic curvature. If $n^\mu$ denotes the (future-directed) unit normal to $\Sigma$ then the extrinsic curvature is defined as
\be
K_{ij}=-\frac12 {\cal L}_n \gamma_{ij}
\ee
We say that the initial data set is asymptotically flat if it satisfies the following three requirements. First of all, there exists a compact subset $S$ of $\Sigma$ such that $\Sigma\backslash S$ is the disjoint union of a finite number of open sets called ends, each of which is diffeomorphic to the complement of a closed ball in $\mathbb{R}^3$. Secondly, in a suitable coordinate system $\gamma_{ij}-\delta_{ij}$ and $K_{ij}$ approach zero at a suitable rate (made more precise in the next subsection) as $r\equiv \sqrt{x^ix^i}\to \infty$. Furthermore, we require that the initial data set satisfies the constraint equations of the Einstein-matter equation (with the convention $16 \pi G = 1$)
\be
G_{\mu\nu}=\frac12 T_{\mu\nu}
\ee
with some energy-momentum tensor $T_{\mu\nu}$: the Hamiltonian constraint
\be\label{gr_ham}
{\cal H}\equiv \frac12\left(R[D]+\gamma^{i_1i_2}_{j_1j_2}K_{i_1}{}^{j_1}K_{i_2}{}^{j_2}-\varrho\right)=0
\ee
and the momentum constraint
\be\label{gr_mom}
{\cal M}^i\equiv \gamma^{ii_1}_{jj_1}D^jK_{i_1}{}^{j_1}-\frac12 \mathfrak{p}^i=0.
\ee
In these equations, $D$ is the covariant derivative associated with $\gamma$, $R[D]$ is the corresponding Ricci scalar, $\varrho$ and $\mathfrak{p}^i$ are the energy and momentum densities of $T_{\mu\nu}$, i.e.
\be
\varrho \equiv T_{\mu\nu}n^\mu n^\nu,\qquad \qquad \mathfrak{p}^i\equiv T_{\mu\nu}n^\mu\gamma^{i\nu},
\ee
and finally,
\be
\gamma^{i_1...i_n}_{j_1...j_n}\equiv n! \gamma^{i_1}_{[j_1}\ldots \gamma^{i_n}_{j_n]}
\ee

To construct an initial data set obeying the above conditions, we follow the recipe of York and perform a conformal transformation on the metric $\gamma$, with conformal factor $\psi$. The conformal metric will be denoted by $\tilde \gamma$:
\be\label{conf_metric}
\gamma_{ij}=\psi^{4}\tilde \gamma_{ij}.
\ee
The inverse of $\tilde \gamma_{ij}$ will be denoted by $\tilde\gamma^{ij}$ so that $\tilde \gamma^{ij}=\psi^4 \gamma^{ij}$. In particular, note that $\tilde\gamma^{ij}\neq \tilde\gamma_{kl}\gamma^{ik}\gamma^{jl}$! Henceforth, indices of tensor fields whose notations contain a tilde will be raised and lowered with $\tilde\gamma^{ij}$ and $\tilde \gamma_{ij}$, respectively. On the other hand, indices of tensor fields denoted by letters without a tilde will be raised and lowered using $\gamma$.

It is useful to decompose $\tilde K_{ij}$ to a trace part and a traceless part as
\be\label{extr_K_trace}
K_{ij}=\psi^{-2}\tilde A_{ij}+\frac13 K \gamma_{ij}
\ee
with $K\equiv \gamma^{ij}K_{ij}$. Using the new variables $(\psi, \tilde \gamma, \tilde A, K)$, the Hamiltonian constraint can be rewritten as
\be\label{Ham_GR_v2}
-\frac14\psi^5{\cal H}\equiv \tilde \gamma^{kl}\tilde D_k\tilde D_l \psi-\frac18 \psi  R[\tilde D]-\frac{1}{12}\psi^5 K^2+\frac18\psi^{-7}\tilde A_{kl}\tilde A^{kl}+\frac18 \psi^{-3}\tilde \varrho=0
\ee
This conformal version of the Hamiltonian constraint is called the Lichnerowicz equation. One can also rewrite the momentum constraint using the new variables as
\be\label{mom_GR_v2}
\psi^{10}{\cal M}^i\equiv \tilde D_j \tilde A^{ij}-\frac23 \psi^6 \tilde \gamma^{ij}\tilde D_j K-\frac12 \tilde p^i=0.
\ee
Here, $\tilde D$ is the covariant derivative associated with $\tilde \gamma$, $R[\tilde D]$ is the corresponding Ricci scalar and we additionally introduced the conformally rescaled energy and momentum density
\be\label{conf_rho_p}
\tilde \varrho \equiv \psi^8 \varrho \qquad\qquad \tilde{ \mathfrak{p}}^i\equiv \psi^{10} \mathfrak{p}^i.
\ee
There are two main reasons for choosing this particular scaling. The first one is a mathematical reason: this scaling is well-suited for proofs of existence and uniqueness of the constraint system. Secondly, it has some physical significance as well: if $(\tilde \varrho, \tilde{ \mathfrak{p}}^i)$ satisfies the dominant energy condition
$$\tilde \varrho \geq \sqrt{\tilde \gamma_{ij}\tilde{ \mathfrak{p}}^i\tilde{\mathfrak{p}}^j} $$
then so does the physical energy and momentum densities, i.e.
$$\varrho \geq \sqrt{ \gamma_{ij} \mathfrak{p}^i \mathfrak{p}^j}  $$
If the matter source is a real scalar field then it is useful to define
\be\label{def_Kphi}
K_\phi\equiv-\frac12\mathcal{L}_n\phi, \qquad \qquad \tilde K_\phi\equiv \psi^6 K_\phi
\ee
in analogy with the extrinsic curvature. For a minimally coupled scalar field (theory \eqref{4dST} with $\epsilon_1,\epsilon_2=0$), one can express the energy and momentum densities of the scalar field in terms of the conformal variables 
\bea
\varrho&=&\frac12\left(4\psi^{-12}\tilde K_\phi^2+\psi^{-4}\tilde \gamma^{ij}\partial_i\phi\partial_j\phi\right)+V(\phi) \label{scalar_rho} \\
\mathfrak{p}^i&=&2\psi^{-10} \tilde K_\phi \tilde \gamma^{ij}\partial_i\phi \label{scalar_p}
\eea
We can see that in the expression for $\varrho$ there are terms with different conformal scaling and hence for a scalar field it is not necessarily beneficial to use the variable $\tilde \varrho$ of equation \eqref{conf_rho_p}. However, $\mathfrak{p}^i$ is York-scaled with
\be
\tilde{ \mathfrak{p}}^i=2 \tilde K_\phi \tilde \gamma^{ij}\partial_i\phi.
\ee
To solve the momentum constraint, York proposed a decomposition of the conformal extrinsic curvature $\tilde A_{ij}$ to a longitudinal and transverse-traceless (TT) part
\be\label{ctt_A}
\tilde A^{ij}=\tilde A^{ij}_{\rm TT}+\tilde A^{ij}_{\rm L}
\ee
To make it clear, the tracelessness and the transversality of $\tilde A_{\rm TT}$ can be expressed as
\be
\tilde \gamma_{ij}\tilde A^{ij}_{\rm TT}=0 \qquad \qquad \tilde D_i \tilde A^{ij}_{\rm TT}=0
\ee
As explained below, the longitudinal piece $\tilde A_L$ can be expressed as the action of the {\it conformal Killing operator} $\tilde L$ on a vector field $Y^i$, that is,
\be\label{conf_KVF}
\tilde A_L^{ij}=(\tilde L Y)^{ij}\equiv \tilde D^i Y^j+\tilde D^j Y^i-\frac23 \tilde \gamma^{ij}\tilde D_kY^k
\ee
The existence and uniqueness of such a decomposition hinges on the existence and uniqueness of the solution to the equation
\be
\tilde \Delta_L Y^i=\tilde D_j\tilde A^{ij}
\ee
where the differential operator $\tilde\Delta_L$ is called the {\it conformal vector Laplacian} and it can be defined with its action on a vector field
\be\label{conf_vector_laplace}
\tilde\Delta_L Y^i\equiv \tilde D_j (\tilde L Y)^{ij} = \tilde D_j \tilde D^j Y^i+\frac13 \tilde D^i \tilde D_j Y^j+\tilde R^i{}_j Y^j.
\ee
In terms of the variables $(Y,\tilde \gamma, \tilde K)$, the momentum constraint is
\be\label{gr_conf_mom}
\tilde \Delta_L Y^i-\frac23 \psi^{6}\tilde \gamma^{ij}\tilde D_j K=\frac12 \tilde{ \mathfrak{p}}^i
\ee
We can see that the advantage of the variables $(\psi,\tilde \gamma, \tilde A_{\rm TT}, K, Y)$ is that the Hamiltonian and the momentum constraints take the form of a system of four elliptic partial differential equations consisting of \eqref{Ham_GR_v2} and \eqref{gr_conf_mom}. These equations are to be solved for the variables $(\psi, Y)$. The rest of the variables $(\tilde\gamma, \tilde A_{\rm TT}, K)$ play the role of sources to be chosen freely.

In general relativity, it has been demonstrated \cite{Choquet-Bruhat:2009xil} that the above procedure actually leads to a unique asymptotically flat initial data set for a suitable choice of $(\tilde \gamma, \tilde A_{\rm TT}, K)$. In the next subsection, we review some theorems (for both vacuum gravity and gravity minimally coupled to a scalar field) containing conditions on the free data that guarantee the existence of a unique solution for $(\psi, Y)$ with the desired regularity and asymptotic fall-off.

\subsection{Mathematical results in General Relativity}\label{sec:gr_ctt_math}

We continue the discussion by stating some of the previous statements in a mathematically more precise form (see e.g. \cite{Choquet-Bruhat:2009xil} and references therein). We begin with a brief discussion of weighted Sobolev spaces that capture the desired asymptotic fall-off requirements. In this section, we consider initial data surfaces of $n\geq 3$ dimensions.

Given coordinates $x^i$ on the initial slice $\Sigma$, let $\sigma(x)\equiv (1+|x|^2)^{1/2}$ with $|x|\equiv \sqrt{\delta_{ij}x^ix^j}$ and let $1\leq p \leq \infty$, $\delta\in \mathbb{R}$, $s\in \mathbb{N}$. Given a tensor field $u$ of some given type on $\Sigma$, one may define the weighted Sobolev norm
\be
|| u||_{p,s,\delta}\equiv \sum\limits_{0\leq |m|\leq s}||\sigma^{|m|+\delta}(x) D^m u||_{L^p(\Sigma)}.
\ee
Then we say that the weighted Sobolev space $W^p_{s,\delta}(\Sigma)$ is the space of tensor fields on $\Sigma$ whose norm $||\cdot ||_{p,s,\delta}$ is finite. If $\Sigma$ is a $3$-manifold and a tensor field $u$ on $\Sigma$ is in $ W^p_{s,\delta}(\Sigma)$, then it must fall off faster than $r^{-\delta-3/p}$ as $r\to \infty$ and similarly its derivatives $D^m u$ must fall off faster than $r^{-\delta-3/p-|m|}$. Hence, the role of the parameter $\delta$ is to provide additional information on the the asymptotic behaviour of the tensor field.

Now we can define the notion of an asymptotically Euclidean manifold endowed with a Riemannian metric in a certain weighted Sobolev class.

\begin{definition}
An $n$-dimensional manifold $(M,\gamma)$ is said to be {\it asymptotically Euclidean of class} $W^p_{\sigma,\rho}$ if
\begin{itemize}
    \item[(i)] There exists a finite number of open sets $\{E_I\}$ called ends, a compact set $S$ and a set of diffeomorphisms $\{\Phi_I\}$ such that $\Sigma\backslash S=\bigcup\limits_{I} E_I$ and $\Phi_I$ maps each $E_I$ to a complement of a closed ball in $\mathbb{R}^n$.
    \item[(ii)] In each end $(\Phi_I^\star\gamma)_{ij}-\delta_{ij}\in W^p_{\sigma,\rho}$ with $\sigma>\frac{n}{p}$ and $\rho>-n/p$.
\end{itemize}
\end{definition}

Requirement (ii) in the above definition captures the condition that (in asymptotically Euclidean coordinate system) the components of the metric are $\gamma_{ij}=\delta_{ij}+{\cal O}(r^{-\lambda})$ as $r\to \infty$ where $\lambda$ is an arbitrary positive real number. This is a less restrictive condition than the usual assumption of $\gamma_{ij}=\delta_{ij}+{\cal O}(r^{-1})$. In the following lemma we collect some of the useful properties of weighted Sobolev spaces \cite{cantor1977,CM_1981__43_3_317_0}.

\begin{lemma}\label{w_sobolev_properties}
Properties of weighted Sobolev spaces.
\begin{itemize}
    \item[(i)] Let $s>\frac{n}{p}+k$. Then $W^p_{s,\delta}(\mathbb{R}^n)\subset C^k(\mathbb{R}^n)$.
    \item[(ii)] Pointwise multiplication satisfies $W^p_{s_1,\delta_1}(\Sigma)\times W^p_{s_2,\delta_2}(\Sigma) \to W^p_{s_3,\delta_3}(\Sigma)$ provided that $s_3<s_1+s_2-\frac{n}{p}$ and $\delta_3<\delta_1+\delta_2+\frac{n}{p}$.
    
    Moreover, let $p> 1$, $s>\frac{n}{p}$, $0\leq l \leq s$ and define
    \be\label{def:sob_1}
    W^p_{s,\delta}(1,\Sigma)=\left\{ f:\Sigma\to\mathbb{R} ~:~ f-1\in W^p_{s,\delta}(\Sigma)\right\}.
    \ee
    Then pointwise multiplication induces smooth maps
    \bea
    W^p_{s,\delta}(1,\Sigma)\times W^p_{s-l,\delta+l}(\Sigma) \to W^p_{s-l,\delta+l}(\Sigma) \\ 
    W^p_{s,\delta}(1,\Sigma)\times W^p_{s-l,\delta+l}(1,\Sigma) \to W^p_{s-l,\delta+l}(1,\Sigma)
    \eea
\end{itemize}
\end{lemma}

Part (i) of this lemma just tells us that a sufficiently weighted Sobolev-regular tensor field are ($k$ times) continuously differentiable. The rough version of the second part of the lemma implies the following. Given a tensor field $T$ which is expressed as a contraction (pointwise multiplication) of several sufficiently regular tensor fields $S_A$ (i.e. $S_A\in W^p_{s_A,\delta_A}(\Sigma)$ and $s_A$ is large enough). Then the weighted Sobolev regularity of $T$ is the same as that of the least regular factor $S_A$.

Now let us turn to the actual discussion of the conformally formulated constraint equations and state some previous results. We start with the momentum constraint.

\begin{theorem}\label{gr_ctt_momentum}
Let $(\Sigma,\gamma)$ be a $W^p_{s,\delta}$ asymptotically Euclidean manifold with $p>\frac{n}{2}$, $s>\frac{n}{p}$ and $-\frac{n}{p}<\delta <n-2-\frac{n}{p}$. Assume that we are given $\tilde A_{\rm TT}\in W^p_{s-1,\delta+1}(\Sigma)$.
\begin{itemize}
    \item[(i)] Consider the case of maximal slicing $K\equiv 0$. Then the constraints decouple and the momentum constraint has a unique solution $Y\in W^p_{s,\delta}(\Sigma)$.
    \item[(ii)]  Consider the case of nonzero $K\in W^p_{s-1,\delta+1}(\Sigma)$ and suppose that we are given $\psi>0$, $\psi\in W^p_{s,\delta}(1,\Sigma)$. Then the momentum constraint has a unique solution $Y\in W^p_{s,\delta}(\Sigma)$.
\end{itemize}

\end{theorem}

Next we turn to the Lichnerowicz equation which is highly non-linear and therefore more subtle. We specifically concentrate on the vacuum case and assume that $\Sigma$ is a maximal slice. Similar results hold when $\Sigma$ has (nearly or exactly) constant mean curvature. Before stating the theorem on the Lichnerowicz equation, we need to discuss a topological condition on the initial data surface to ensure uniqueness.

Consider an $n$-dimensional asymptotically Euclidean manifold $(\Sigma,\gamma)$ of class $W^p_{s,\delta}(\Sigma)$. We say that such a manifold is in the {\it positive Yamabe class} if the functional
\be\label{yamabe_inv_def}
    I_\gamma[f]\equiv\int_\Sigma d^n x\sqrt{\gamma} \left[ \gamma^{ij}\partial_i f\partial_j f+\frac18 R[\tilde D]f^2\right]>0
\ee
is positive for every nontrivial function $f$ of type $W^p_{2,\rho}(\Sigma)$ with $\rho>-\frac{n}{p}+\frac{n}{2}-1$. Interestingly, the functional $I_\gamma[f]$ is conformally invariant in the sense that
\be
I_\gamma[f]=I_{\gamma^\prime}[f^\prime], \qquad \qquad \gamma^\prime = \theta^{\tfrac{4}{n-2}}\gamma, \qquad f^\prime = \theta^{-1} f
\ee

\begin{theorem}\label{gr_lichnerowicz}
We make the following hypotheses.
\begin{itemize}
    \item[(i)] Let $(\Sigma,\tilde \gamma)$ be an $n$-dimensional asymptotically Euclidean manifold of type $W^p_{s,\delta}$ with
    $$p>\frac{n}{2},\qquad s>\frac{n}{p} \qquad {\rm and} \qquad -1+\frac{n}{2}-\frac{n}{p}<\delta<-2+n-\frac{n}{p}.$$ 
    \item[(ii)] Assume that we are given $\tilde A_{TT} \in W^p_{s-1,\delta+1}(\Sigma)$ and $K=0$ as free data.
    \item[(iii)] Suppose that $(\Sigma, \tilde \gamma)$ is in the positive Yamabe class.

    \item[(iv)] Finally, assume that we are given a vector field $Y\in W_{s,\delta}^p(\Sigma)$.
\end{itemize}
Under hypotheses (i)-(iv) the Lichnerowicz equation has a unique solution $\psi>0$ with $\psi\in W^p_{s,\delta}(1,\Sigma)$.
\end{theorem}

Note that there are different versions of the above theorem which do not require the data to be of the positive Yamabe class but in those cases the other conditions are more restrictive (see \cite{Choquet-Bruhat:2009xil} for a more complete account).

Consider now the case when we include a minimally coupled scalar field with nonnegative potential $V(\phi)\geq 0$. As noted in \eqref{scalar_rho}, $\varrho$ contains terms with different conformal scalings, introducing "bad" terms in the Lichnerowicz equation. Nevertheless, one still has the following theorem \cite{Choquet2000}.

\begin{theorem}\label{thm_esf_lichnerowicz}
Let $(\Sigma,\gamma)$ be a $3$-dimensional asymptotically Euclidean manifold of class $W^p_{s,\delta}$ with $K=0$, $s>2+\tfrac{3}{p}$ and $\delta>-\tfrac{3}{p}$. We further require that 
$$R[\tilde D]-\frac12\tilde \gamma^{ij}\partial_i\phi\partial_j\phi>0.$$
Then there exists an open set of values for the free data $(\tilde A_{TT},\phi,\tilde K_\phi)$ satisfying $\phi-\phi_\infty\in W^p_{s,\delta}(\Sigma)$ where $\phi_\infty$ is the asymptotic (constant) value of the scalar field, $\tilde A_{TT},\tilde K_\phi \in W^p_{s-1,\delta+1}(\Sigma)$ and $V(\phi)\in W^p_{s-2,\delta+2}(\Sigma)$ such that the Lichnerowicz equation has a solution $\psi>0$ of class $W^p_{s,\delta}(1,\Sigma)$.
\end{theorem}

It is worth pointing out that even though Theorem \ref{thm_esf_lichnerowicz} allows for $\phi$ to have a non-zero asymptotic value, the requirement $V(\phi)\in W^p_{s-2,\delta+2}(\Sigma)$ implies that we must have $V(\phi_\infty)=0$. Of course, in general, $\phi\in W^p_{s,\delta}$ does not imply $V(\phi)\in W^p_{s-2,\delta+2}$. In fact, this condition puts an implicit constraint on $V$. However, it also shows that $V$ does not need to be a smooth function.

Finally, we state the corresponding results for the full CTT system of constraints. In the vacuum case, assuming $K=0$, the Hamiltonian and momentum constraints decouple so combining theorems \ref{gr_ctt_momentum} and \ref{gr_lichnerowicz} leads to

\begin{theorem}\label{thm_gr_constraints}
Suppose that the hypotheses (i)-(iii) of Theorem \ref{gr_lichnerowicz} hold. Then the CTT system of constraint equations (\eqref{Ham_GR_v2} and \eqref{gr_conf_mom} with $\tilde\varrho,\tilde{ \mathfrak{p}}=0$) admits a unique solution $\psi\in W^p_{s,\delta}(1,\Sigma)$, $Y\in W^p_{s,\delta}(\Sigma)$ with $\psi>0$.
\end{theorem}

Similarly, in the case of gravity minimally coupled to a scalar field, we have
\begin{theorem}\label{thm_esf_constraints}
Let $(\Sigma,\gamma)$ be a $3$-dimensional asymptotically Euclidean manifold of class $W^p_{s,\delta}$ with $K=0$, $p>\tfrac32$, $s>2+\tfrac{3}{p}$ and $1-\tfrac{3}{p}>\delta>-\tfrac{3}{p}$. Moreover, suppose
$$R[\tilde D]-\frac12\tilde \gamma^{ij}\partial_i\phi\partial_j\phi>0.$$
Then there is an open set of values for $(\tilde A_{TT},\phi,\tilde K_\phi)$ satisfying $\phi-\phi_\infty\in W^p_{s,\delta}(\Sigma)$, $\tilde A_{TT},\tilde K_\phi \in W^p_{s-1,\delta+1}(\Sigma)$ and $V(\phi)\in W^p_{s-2,\delta+2}(\Sigma)$ such that the conformally formulated constraints have a solution $(\psi,Y)$ with $\psi\in W^p_{s,\delta}(1,\Sigma)$, $\psi>0$ and $Y\in W^p_{s,\delta}(\Sigma)$.
\end{theorem}

\subsection{Bowen-York initial data in General Relativity}\label{sec:gr_by}

The theorems presented in the previous section provide the mathematical underpinning of the construction of initial data in General Relativity (in vacuum and with a minimally coupled scalar field). It remains to find a suitable choice of free data such that the corresponding solution represents physical systems of interest. A proposal for the choice of free data that yields a slice of a spacetime with multiple black holes (in vacuum) was originally put forward by Bowen and York \cite{Bowen:1980yu}. Their approach was later modified to better suit for numerical relativity purposes, see e.g. \cite{Brandt:1997tf}. We briefly review this latter approach.

First of all, let the initial hypersurface $\Sigma$ be $\mathbb{R}^3$ with a puncture at the origin, i.e. $\Sigma=\mathbb{R}\backslash\{O\}$. Furthermore, consider the simple choice
\be
\tilde \gamma_{ij}=\delta_{ij}, \qquad K=0, \qquad \tilde A^{ij}_{TT}=0
\ee
and assume that we are in vacuum, i.e. $\tilde \varrho=0$, $\tilde{ \mathfrak{p}}^i=0$. Then as shown in \cite{Bowen:1980yu} one can obtain the following $6$-parameter family of analytical solutions to the momentum constraint
\be
Y^i=-\frac{1}{4r}\Bigr(7P^i+P^j\hat x_j \hat x^i\Bigl)-\frac{1}{r^2}\epsilon^{ijk}S_j\hat x_k
\ee
where $x^i$ are the usual Euclidean coordinates on $\mathbb{R}^3$, $r=\sqrt{\delta_{ij}x^ix^j}$ is the Euclidean distance from the origin, $\hat x^i\equiv x^i/r$ and $\epsilon^{ijk}$ is the Levi-Civita tensor associated with the flat metric. The vectors $P^i$ and $S^i$ are to be chosen freely. The corresponding conformal extrinsic curvature is
\be\label{one_puncture}
{\tilde A}_{ij}{(P,S)}=\frac{3}{2 r^2}\Bigl(P_i\hat x_j+P_j\hat x_i-(\gamma_{ij}-\hat x_i\hat x_j)P^k\hat x_k\Bigr)+\frac{3}{r^3}\left(\epsilon_{ikl}S^k {\hat x}^l {\hat x}_j+\epsilon_{jkl}S^k {\hat x}^l {\hat x}_i\right)
\ee
Interestingly, the components of the ADM linear momentum and the canonical angular momentum of this data are
\be
P_{\rm ADM}^i=\frac{1}{8\pi}\int\limits_{S_\infty}\mathrm{d}A~ (K^{ij}-K\gamma^{ij})\hat x_j=P^i
\ee
and
\be
J^i=\frac{1}{8\pi}\int\limits_{S_\infty}\mathrm{d}A~ (K_{jk}-K\gamma_{jk})\epsilon^{ijl} x_l\hat x^k=S^i,
\ee
respectively. Note that since the momentum constraint is linear in $Y$ and the solution \eqref{one_puncture} is linear in $(P,S)$, a linear superposition of any number of such solutions is also a solution to the momentum constraint. More precisely, if the initial slice is $\mathbb{R}^3$ with $N$ punctures at coordinate locations $c^i_{(\alpha)}$ then
\bea
{\tilde A}^{ij}&=\sum\limits_{\alpha=1}^N&\biggr[\frac{3}{2 r_{(\alpha)}^2}\Bigl(P^i_{(\alpha)}\hat x^j_{(\alpha)}+P^j_{(\alpha)}\hat x^i_{(\alpha)}-\left(\delta^{ij}-\hat x^i_{(\alpha)}\hat x^j_{(\alpha)}\right)\Bigr) \nonumber \\
& &+\frac{3}{r^3_{(\alpha)}}\left(\epsilon^{i}{}_{kl}S^k_{(\alpha)} {\hat x}^l_{(\alpha)} {\hat x}^j_{(\alpha)}+\epsilon^j{}_{kl}S^k_{(\alpha)} {\hat x}^l_{(\alpha)} {\hat x}^i_{(\alpha)}\right)\biggr] \label{gr_by_extr_K}
\eea
with $r_{(\alpha)}\equiv |x-c_{(\alpha)}|$ and $\hat x_{(\alpha)}\equiv \left(x-c_{(\alpha)}\right)/r_{(\alpha)}$ is an exact solution of \eqref{gr_conf_mom} (provided that $K=0$ and $\tilde p^i=0$). Clearly, in this case $P_{(\alpha)}$ and $S_{(\alpha)}$ represent the ADM linear and angular momenta of the black holes in case of large separation.

Although the momentum constraint has an exact solution, the Hamiltonian constraint is non-linear and there is no closed formula known for the corresponding solution. However, one can seek the solution in the form
\be\label{by_psi}
\psi=1+\frac{1}{\mu}+u, \qquad \qquad \frac{1}{\mu}\equiv\sum\limits_{\alpha=1}^N\frac{m_{(\alpha)}}{2|x-c_{(\alpha)}|}
\ee
where the parameters $m_{(\alpha)}$ are called the bare masses of the punctures. The significance of the ansatz \eqref{by_psi} is that $\psi$ is separated to two pieces: $1/\mu$ is singular at the punctures, whereas $u$ turns out to be regular in the {\it entire} $\mathbb{R}^3$.

To show that the solution for $u$ is indeed regular, one needs to formulate an elliptic boundary value problem for $u$ on $\mathbb{R}^3$. This amounts to rewriting the Lichnerowicz equation in terms of $u$ and complementing it with a boundary condition at $r\to \infty$. It is worth emphasizing that since this problem is solved for $u$ on $\mathbb{R}^3$ without excising the punctures, no interior boundary conditions are required for $u$. Since the flat Laplacian annihilates $1+1/\mu$ (away from the punctures), the boundary value problem to be solved for $u$ is 
\bea
\Delta u+\frac18 \mu^7 {\tilde A}_{ij}{\tilde A}^{ij}\Bigl(1+\mu (1+u)\Bigr)^{-7}&=0 \label{by_elliptic} \\
\lim\limits_{r\to\infty}\partial_r(ru)&=0
\eea
where ${\tilde A}^{ij}$ is given by \eqref{gr_by_extr_K}. The boundary condition for $u$ guarantees that $u=\mathcal{O}(r^{-1})$ as $r\to \infty$ which corresponds to the original notion of asymptotic flatness.

As discussed in \cite{Brandt:1997tf}, this elliptic boundary value problem admits a unique $C^2$ solution\footnote{In more detail, the solution is $u\in W^p_{3,\delta}(\mathbb{R}^3)$ with $p>3$ and $0\leq \delta <1-\tfrac{3}{p}$. Lemma \ref{w_sobolev_properties} then guarantees that $u$ is $C^2$.} on $\mathbb{R}^3$ (including the punctures). An important part of the proof is to observe that $\mu^7\tilde A_{ij}\tilde A^{ij}$ scales as $|x-c_{(\alpha)}|$ and therefore, in the elliptic equation \eqref{by_elliptic} the nonlinear term has a continuous prefactor. 

The solution represents a compactification of $N+1$ asymptotically flat ends where each puncture corresponds to spatial infinity. Another interesting property of the solution is that in the post-Newtonian (PN) approximation of the two-body point mass system the leading order contribution to the conjugate momentum is exactly of the Bowen-York form, provided that one uses the so-called ADMTT gauge condition \cite{Tichy:2002ec}. This means that up to third order in the velocities (i.e. of order $(v/c)^3$) the superposition of two Bowen-York extrinsic curvatures represents a slowly moving binary black hole system with large separations. Note however, that this is not true for finite $P$ and $J$: in particular, a single angular momentum source is not exactly a slice of the Kerr solution, but rather it contains some additional junk radiation. This unphysical gravitational wave content is significant when the separation of the black holes is not large enough or the momenta $P$, $J$ are not small enough. If one wishes to construct initial data for a binary system with relatively small separation, moving along quasicircular orbits, then the puncture data may not be suitable. One way to improve on the data would be to include higher order PN corrections in the free part of the data \cite{Tichy:2002ec}. A different approach to construct data that represents a binary system in quasiequilibrium will be presented later. Nevertheless, the puncture approach serves as a fairly good initial data provided that one can afford to simulate a sufficient number of orbits so that there is time for the junk radiation to disperse.

\subsection{Conformal transverse traceless decomposition for scalar-tensor effective field theory}\label{sec:4dst_ctt}

We now turn our attention to the scalar-tensor theory \eqref{4dST} and extend previous results (reviewed in sections \ref{sec:gr_ctt}-\ref{sec:gr_by}) to these theories.

We start by giving the constraint equations of the theory \eqref{4dST}. (These equations can also be found in e.g. \cite{Berti2020,Witek:2020uzz} for the theory $\epsilon_1= 0$.) The Hamiltonian constraint can be written as
\bea\label{4dst_ham}
2{\cal H}&\equiv & R[D]+\gamma^{i_1i_2}_{j_1j_2}K_{i_1}{}^{j_1}K_{i_2}{}^{j_2}-V(\phi)-\frac12 (D\phi)^2-2K_\phi^2-\epsilon_1 f_1(\phi)X\left(6K_\phi^2+\frac12 (D\phi)^2\right) \nonumber \\
& &-2\epsilon_2\gamma^{i_1i_2i_3}_{j_1j_2j_3}\left(R[D]_{i_1i_2}{}^{j_1j_2}+2K_{i_1}{}^{j_1}K_{i_2}{}^{j_2}\right)\left(D_{i_3}D^{j_3} f_2(\phi)-2 f_2^\prime(\phi)K_\phi K_{i_3}{}^{j_3}\right)=0
\eea
The momentum constraint has a particularly compact form when written in terms of the canonical momenta:
\be\label{4dst_mom_1}
{\cal M}_i\equiv D^j\left(\pi_{ij}\gamma^{-1/2}\right)-\frac12\pi_\phi \gamma^{-1/2} D_i\phi =0
\ee
where the momenta conjugate to $\phi$ and $\gamma_{ij}$ are given by
\bea\label{pi_phi}
\frac{\pi_\phi}{ \gamma^{1/2}} &= &-2K_\phi(1+2\epsilon_1 f_1(\phi) X) -2\epsilon_2 f_2^\prime (\phi)\gamma^{i_1i_2i_3}_{j_1j_2j_3}K_{i_1}{}^{j_1} \left(R[D]_{i_2i_3}{}^{j_2j_3}+\frac23 K_{i_2}{}^{j_2}K_{i_3}{}^{j_3}\right)
\eea
and
\bea\label{pi_ij}
\frac{\pi^i{}_j}{\gamma^{1/2}}&=&\gamma^{ii_1}_{jj_1}K_{i_1}{}^{j_1}+2\epsilon_2\gamma^{ii_1i_2}_{jj_1j_2}\biggl[2K_{i_1}{}^{j_1}D_{i_2}D^{j_2} f_2(\phi)- f_2^\prime(\phi)K_{\phi} \left(R[D]_{i_1i_2}{}^{j_1j_2}+2 K_{i_1}{}^{j_1}K_{i_2}{}^{j_2}\right)\biggr]
\eea
respectively. After some algebraic manipulations, the momentum constraint can also be written as
\bea\label{4dst_mom_2}
{\cal M}^i&\equiv & \gamma^{ii_1}_{jj_1}D^j K_{i_1}{}^{j_1}+K_\phi(1+2\epsilon_1 f_1(\phi) X) D^i\phi \nonumber \\
& & +\epsilon_2\gamma^{ii_1i_2}_{jj_1j_2}\Bigl\{4\left(D^jK_{i_1}{}^{j_1}\right)\left(D_{i_2}D^{j_2}f_2(\phi)-2f_2^\prime(\phi)K_\phi K_{i_2}{}^{j_2}\right) \nonumber \\
& & +\left[K_k{}^jD^k f_2(\phi)-2D^j\left( f_2^\prime(\phi)K_\phi\right)\right]\left(R[D]_{i_1i_2}{}^{j_1j_2}+2 K_{i_1}{}^{j_1}K_{i_2}{}^{j_2}\right)\Bigr\}=0
\eea

The next step is to follow the recipe of section \ref{sec:gr_ctt} and write the constraint equations \eqref{4dst_ham}-\eqref{4dst_mom_1} in terms of the variables $(\psi, \tilde \gamma_{ij}, K, \tilde A^{\rm TT}_{ij},Y^i,\phi, \tilde K_{\phi})$. Clearly, the form of the CTT constraint equations is not quite as elegant as in General Relativity: the terms coming from the $4$-derivative corrections have different conformal scalings. Nevertheless, the point is that at sufficiently weak couplings, the conformally formulated constraints will constitute a system of elliptic PDEs with a well-posed boundary value problem for the variables $(\psi,Y)$.

By weak couplings, we simply mean that for any (smooth) choice of the functions $f_1,f_2$, there exists an open set of values for $(\epsilon_1,\epsilon_2)$ in a neighborhood of $(0,0)$ such that constraint equations \eqref{4dst_grav_eom}-\eqref{4dst_scalar_eom} yield a solution. In particular, this means that the $4$-derivative corrections are small compared to the Einstein-minimally coupled scalar terms, or in other words,
\be\label{small}
   \epsilon_1 f_1(\phi), \epsilon_1 f_1'(\phi), \epsilon_2 f_2'(\phi), \epsilon_2 f_2''(\phi)\ll\Lambda^{2}
\ee
where $\Lambda$ is any length scale defined by the Riemann tensor and the first and second derivatives of the scalar field.

To discuss the puncture data for $4\partial$ST theories (later in sections \ref{sec:4dst_by} and \ref{sec:4dst_by_math}), it is useful to conformally rescale the canonical momenta. We introduce
\be
 \pi_{ij}\gamma^{-1/2}\equiv \psi^{-2} \tilde\pi_{ij}+\frac13 \pi \gamma_{ij}, \qquad \qquad  \pi_\phi \gamma^{-1/2} \equiv \psi^{-6} \tilde \pi_\phi 
\ee
The momentum constraint now takes the same form as in general relativity (c.f. \eqref{mom_GR_v2}), after making the substitutions $\tilde A_{ij}\to -\tilde\pi_{ij} $, $K\to \frac12 \pi$ and $\tilde{ \mathfrak{p}}^i\to -\tilde \pi_\phi \tilde D^i\phi$:
\be
-\tilde D_j \tilde \pi^{ij}-\frac13\psi^6\tilde \gamma^{ij}\tilde D_j \pi+\frac12 \tilde \pi_\phi \tilde \gamma^{ij}\tilde D_j\phi=0
\ee
As a reference, we give the Lichnerowicz equation (Hamiltonian constraint in terms of the conformal variables) for $4\partial$ST theories:
\bea\label{4dst_ham_ctt}
0&= & \tilde \gamma^{kl}\tilde D_k\tilde D_l \psi-\frac18 \psi \left(R[\tilde  D]-\frac12 \tilde \gamma^{ij}\partial_i\phi\partial_j\phi\right)-\psi^5 \left(\frac{1}{12}K^2-\frac18 V(\phi)\right)+\frac18\psi^{-7}\left(\tilde A_{kl}\tilde A^{kl}+2\tilde K_\phi^2\right) \nonumber \\
& &+\frac18\epsilon_1 f_1(\phi)\psi^5\left(12\psi^{-24}\tilde K_\phi^4-2\psi^{-16}\tilde K_\phi^2 \tilde\gamma^{ij}\partial_i\phi\partial_j\phi-\frac14 \psi^{-8}(\tilde\gamma^{ij}\partial_i\phi\partial_j\phi)^2\right) \nonumber \\
& & +\frac14 \epsilon_2 \psi^{-3} \biggl[\tilde D_{i_3}\tilde D^{j_3} f_2(\phi)-2\tilde D_{i_3}\ln\psi\tilde D^{j_3} f_2(\phi)-2\tilde D_{i_3} f_2(\phi)\tilde D^{j_3}\ln\psi \nonumber \\
& &+2\tilde\gamma^{j_3}_{i_3}\tilde \gamma^{kl}\tilde D_k\ln\psi\tilde D_l  f_2(\phi)-2\psi^{-8} f_2^\prime(\phi)\tilde K_\phi \tilde A_{i_3}{}^{j_3}-\frac23\psi^{-2}  f_2^\prime(\phi)\tilde K_\phi K \tilde\gamma_{i_3}^{j_3}\biggr] \nonumber \\
& & \times \biggl(\gamma^{i_1i_2i_3}_{j_1j_2j_3} R[\tilde D]_{i_1i_2}{}^{j_1j_2}+8\psi^{-1}\tilde D^{i_3}\tilde D_{j_3}\psi-8\psi^{-1}\tilde\gamma^{i_3}_{j_3}\tilde D_{k}\tilde D^{k}\psi-24\psi^{-2}\tilde D^{i_3}\psi\tilde D_{j_3}\psi \nonumber \\
& & -16\psi^{-2}\tilde\gamma^{i_3}_{j_3}\tilde D_{k}\psi\tilde D^{k}\psi + \frac49\psi^4 K^2\tilde\gamma^{i_3}_{j_3}-\frac43 \psi^{-2} K \tilde A^{i_3}_{j_3}+4\psi^{-8}\tilde A^{i_3}{}_k\tilde A^k{}_{j_3}-2\psi^{-8}\tilde A^{k}{}_l\tilde A^l{}_{k} \tilde\gamma^{i_3}_{j_3}\biggr)
\eea
Similarly, one can write the momentum constraint using the CTT variables
\bea\label{4dst_mom_ctt}
0&= &\tilde \Delta_L Y^i-\frac23 \psi^{6}\tilde \gamma^{ij}\tilde D_j K-\tilde K_\phi\left(1+2\epsilon_1 f_1(\phi)X \right) \tilde \gamma^{ij}\tilde D_j\phi \nonumber \\
& & +4\epsilon_2\psi^{-4}\gamma^{ii_1i_2}_{jj_1j_2}\biggl[\tilde D_{i_2}\tilde D^{j_2} f_2(\phi)-2\tilde D_{i_2}\ln\psi\tilde D^{j_2} f_2(\phi)-2\tilde D_{i_2} f_2(\phi)\tilde D^{j_2}\ln\psi \nonumber \\
& &+2\tilde\gamma^{j_2}_{i_2}\tilde \gamma^{kl}\tilde D_k\ln\psi\tilde D_l  f_2(\phi)-2\psi^{-8} f_2^\prime(\phi)\tilde K_\phi \tilde A_{i_2}{}^{j_2}-\frac23\psi^{-2}  f_2^\prime(\phi)\tilde K_\phi K \delta_{i_2}^{j_2}\biggr] \nonumber \\
& &\times \left(\tilde D^j {\tilde A}_{i_1}{}^{j_1}+\frac13 \tilde\gamma_{i_1}^{j_1} \psi^{6}\tilde D^j K-4{\tilde A}_{i_1}^{j_1}{\tilde D}^j\ln\psi-2\tilde \gamma_{i_1}^{j_1}{\tilde A}^{jl}\tilde D_l\ln\psi\right) \nonumber \\
& & +\epsilon_2 \psi^{-4}\left[{\tilde A}_k{}^j{\tilde D}^k f_2(\phi)+\frac13 \psi^6 K{\tilde D}^j  f_2(\phi)-\psi^6{\tilde D}^j\left( f_2^\prime(\phi)\psi^{-6}{\tilde K}_\phi\right)\right] \nonumber \\
& & \times \biggl(\gamma^{ii_1i_2}_{jj_1j_2} R[\tilde D]_{i_1i_2}{}^{j_1j_2}+8\psi^{-1}\tilde D_{i}\tilde D^{j}\psi-8\psi^{-1}\tilde\gamma^{i}_{j}\tilde D_{k}\tilde D^{k}\psi-24\psi^{-2}\tilde D_{i}\psi\tilde D^{j}\psi \nonumber \\
& & -16\psi^{-2}\tilde\gamma^{i}_{j}\tilde D_{k}\psi\tilde D^{k}\psi + \frac{4}{9}\psi^4 K^2\tilde \gamma^{i}_{j}-\frac43 \psi^{-2} K \tilde A^{i}_{j}+4\psi^{-8}\tilde A^{i k}\tilde A_{kj}-2\psi^{-8}\tilde A^{kl}\tilde A_{kl} \tilde\gamma^{i}_{j}\biggr)
\eea
where one has to convert $\tilde A^{ij}$ to the variables $\tilde A_{\rm TT}$, $Y^i$ using \eqref{ctt_A}-\eqref{conf_KVF}. As explained in the subsequent sections, despite the length and lack of elegance of the equations, it is quite simple to extend the existence and uniqueness theorems of section \ref{sec:gr_ctt_math} to equations \eqref{4dst_ham_ctt}-\eqref{4dst_mom_ctt}, at least for small couplings. In particular, numerical relativists should not be discouraged after seeing these equations: at weak coupling they can be solved iteratively, starting from a GR solution.

\subsection{Existence and uniqueness in scalar-tensor effective field theories}\label{sec:4dst_ctt_math}

In this section we establish a well-posedness result for the elliptic boundary value problem of the CTT formulation of the $4\partial$ST constraints. We start by recalling the implicit function theorem in Banach spaces (see e.g. \cite{Nirenberg:2264150}).

\begin{theorem}[Implicit function theorem in Banach spaces]\label{thm:imp}
Let ${\cal B}_x$, ${\cal B}_y$ and ${\cal B}_z$ be Banach spaces and let ${\cal F}$ be a $C^1$ (in the sense of Fr{\'e}chet derivatives) map ${\cal B}_x\times {\cal B}_y\to {\cal B}_z$ such that ${\cal F}(x_0,y_0)=0$ for some $(x_0,y_0)\in {\cal B}_x\times {\cal B}_y$ and the Fr{\'e}chet derivative $y\mapsto{\cal D}_y{\cal F}|_{(x_0,y_0)}(0,y)$ is an isomorphism ${\cal B}_y\to {\cal B}_z$. Then there exists a neighbourhood $\Omega\subset {\cal B}_x$ of $x_0$ and a unique $C^1$ map ${\cal G}: \Omega \to {\cal B}_y$ such that ${\cal G}(x_0)=y_0$ and ${\cal F}(x,{\cal G}(x))=0$ for every $x\in \Omega$. 
\end{theorem}

It is perhaps worth going on a small detour and reviewing the main idea of the proof of this theorem. Let us denote the Fr{\'e}chet derivative (i.e. "linearization") of the functional $\cal F$ at $(x_0,y_0)$ by ${\cal A}$. Note that $\cal A$ is a linear operator ${\cal B}_y\to {\cal B}_z$ and an isomorphism by assumption. The equation ${\cal F}(x,y)=0$ can be rewritten as
\be
{\cal A}y={\cal R}(x,y) \qquad {\rm or} \qquad y={\cal A}^{-1}{\cal R}(x,y)
\ee
with
\be
{\cal R}(x,y)\equiv {\cal A}y-{\cal F}(x,y).
\ee
It can be shown that there is an open ball in ${\cal B}_x$ centered around $x_0$ such that for a fixed $x$ the map ${\cal A}^{-1}{\cal R}(x,\cdot):{\cal B}_y\to{\cal B}_y$ is a contraction map. The theorem then follows by the Banach fixed point theorem, i.e. there is a unique solution $y$ to the equation $y={\cal A}^{-1}{\cal R}(x,y)$ near $y_0$. The practical importance of the fact that ${\cal A}^{-1}{\cal R}(x,y)$ is a contraction map is that the equation can be solved by iterations.

To state and prove a theorem on the $4\partial$ST constraints, we need to make use of a few standard results on elliptic operators in asymptotically Euclidean manifolds \cite{Choquet2000}. Let $(\Sigma,\tilde \gamma)$ be an $n$-dimensional asymptotically Euclidean manifold of class $W^p_{s,\delta}$ with 
\be\label{hyp_delta_p}
n-2-\frac{n}{p}>\delta>-\frac{n}{p}, \qquad \qquad p>\frac{n}{2}.
\ee

To apply the implicit function theorem to the non-linear $4\partial$ST constraints, one needs to study the linearization of these equations. The first lemma will be used for the discussion of the linearized Lichnerowicz equation.

\begin{lemma}\label{poisson_lemma}
Consider a $W^p_{s-2,\delta+2}$ scalar function $c$ on $\Sigma$ and define the Poisson operator 
$$L_P: u\mapsto \tilde \gamma^{kl}\tilde D_k\tilde D_lu-cu$$
for a $W^p_{s,\delta}$ scalar function $u$ on $\Sigma$. Assume furthermore that \eqref{hyp_delta_p} and one of the following two hypotheses holds
\begin{itemize}
    \item[1)] $c\geq 0$ and $s>2+\tfrac{n}{p}$
    \item[2)] $s\geq 2$ and for any nontrivial $f\in W^p_{s,\delta}(\Sigma)$
    \be
    \int\limits_\Sigma d^n\sqrt{\tilde \gamma}\left(\tilde \gamma^{ij}\partial_i f\partial_j f+cf^2\right)>0
    \ee
\end{itemize}
Then the operator $L_P$ is an isomorphism $W^p_{s,\delta}(\Sigma)\to W^{p}_{s-2,\delta+2}(\Sigma)$, i.e. the unique $W^p_{s,\delta}$ solution of the equation $L_Pu=0$ is $u\equiv 0$. 
\end{lemma}

A similar result will be required for the linearized momentum constraint which comes down to a statement on the conformal vector Laplacian \eqref{conf_vector_laplace} \cite{Choquet-Bruhat:2009xil}.

\begin{lemma}\label{conf_laplace_lemma}
On our $n$-dimensional asymptotically Euclidean manifold $(\Sigma,\tilde \gamma)$ of class $W^p_{s,\rho}$, let us assume the inequalities \eqref{hyp_delta_p}.
\begin{itemize}
    \item[(1)] The kernel of the conformal Laplace operator $\Delta_L$ is the space of conformal Killing vector fields of $(\Sigma,\tilde\gamma)$.
    \item[(2)] The manifold $(\Sigma,\tilde \gamma)$ has no conformal Killing vector fields of class $W^p_{s,\delta}$ under the above assumptions. Therefore, $\Delta_L$ is an isomorphism $W^p_{s,\delta}(\Sigma)\to W^p_{s-2,\delta+2}(\Sigma)$.
\end{itemize}
\end{lemma}

Of course, the initial data slice $(\Sigma,\tilde\gamma)$ can have e.g. smooth conformal Killing vector fields, it is simply growth at infinity that rules out the possibility of a conformal Killing vector field that is weighted Sobolev regular with the above requirements.

Finally, we are in the position to combine all these results to a theorem on the conformally formulated $4\partial$ST constraints.

\begin{theorem}\label{thm:4dst_ctt}
Let $(\Sigma,\tilde \gamma)$ be an $3$-dimensional asymptotically Euclidean manifold of type $W^p_{s,\delta}$ with
$$p>\frac{3}{2},\qquad s>2+\frac{3}{p} \qquad {\rm and} \qquad \frac{1}{2}-\frac{3}{p}<\delta<1-\frac{3}{p}.$$
Assume that we are given a solution of the CTT Einstein-scalar-field constraint equations $(\psi_0,Y_0)$ with $\psi_0>0$ and $\psi_0-1,Y_0 \in W^p_{s,\delta}(\Sigma)$ corresponding to free data $(\tilde \gamma,K,\tilde A_{TT},\phi,\tilde K_\phi)$. Our hypothesis on the free data is $\phi-\phi_\infty \in W^p_{s,\delta}(\Sigma)$ where $\phi_\infty$ is the asymptotic (constant) value of the scalar field, $V(\phi)\in W^p_{s-2,\delta+2}(\Sigma)$, $\tilde A_{TT}, \tilde K_\phi \in W^p_{s-1,\delta+1}(\Sigma)$ and $K=0$.

Then the $4\partial$ST constraint system \eqref{4dst_ham_ctt}-\eqref{4dst_mom_ctt} admits a unique solution $(\psi,Y)$ for sufficiently small values of $\epsilon_1$ and $\epsilon_2$ with the same free data, provided that
\begin{itemize}
    \item[(i)] $f_1(\phi)-f_1(\phi_\infty), f_2(\phi)-f_2(\phi_\infty)\in W^p_{s,\delta}(\Sigma)$ and
    \item[(ii)] for every non-trivial function $f\in W^p_{2,\rho}(\Sigma)$ with $\rho>\frac12-\tfrac{3}{p}$ the following inequality holds:
    \be\label{pos_func_scalar}
    \int\limits_\Sigma d^n x\sqrt{\tilde\gamma} \left[\tilde \gamma^{ij}\partial_i f\partial_j f+cf^2\right]>0
    \ee
    with
    \be
    c\equiv \frac18\left( R[\tilde D]+7\psi_0^{-8}\tilde A_{kl}\tilde A^{kl}+14 \psi_0^{-6}\tilde  K_\phi^2-\frac12 \tilde \gamma^{ij}\partial_i\phi\partial_j\phi-5\psi_0^4 V(\phi)\right) \nonumber
    \ee
\end{itemize} 
The $4\partial$ST solution $(\psi,Y)$ is near the Einstein-scalar-field solution in the sense of $W^p_{s,\delta}$ norms.
\end{theorem}

\begin{proof}
The result follows by a straightforward application of the implicit function theorem. Take
\be
x\equiv \bigl(\epsilon_1 f_1(\phi),\epsilon_2 f_2(\phi)\bigr), \qquad y\equiv (\psi, Y), \qquad z\equiv \bigl({\cal H}(x,y),{\cal M}(x,y)\bigr)
\ee
with
\be
{\cal B}_x\equiv W^p_{s,\delta}(\Sigma)\times W^p_{s,\delta}(\Sigma), \qquad {\cal B}_y=W^p_{s,\delta}(\Sigma)\times W^p_{s,\delta}(\Sigma)
\ee

The next step is to use the Sobolev multiplication properties (Lemma \ref{w_sobolev_properties}) together with the assumptions on $s$ and $\delta$ to check that the $4\partial$ST Hamiltonian and momentum constraints \eqref{4dst_ham_ctt}-\eqref{4dst_mom_ctt} map into
\be
{\cal B}_z\equiv W^p_{s-2,\delta+2}(\Sigma)\times W^p_{s-2,\delta+2}(\Sigma)
\ee
To see this, we first note that the $m^{\rm th}$ derivatives of a tensor field of class $W^p_{s,\delta}$ are of class $W^p_{s-m,\delta+m}$. Taking $s_1=s_2=s-1$, $s_3=s-2$, $\delta_1=\delta_2=\delta+1$ and $\delta_3=\delta+2$ in Lemma \ref{w_sobolev_properties}, we find that pointwise multiplication of two tensor fields of class $W^p_{s-1,\delta+1}$ is of class $W^p_{s-2,\delta+2}$. Similarly, with $s>2+3/p$ and $\delta>-3/p$, pointwise multiplication satisfies
\bea
W^p_{s-2,\delta+2}(\Sigma)\times W^p_{s,\delta}(\Sigma)&\to &W^p_{s-2,\delta+2}(\Sigma) \nonumber \\
W^p_{s-2,\delta+2}(\Sigma)\times W^p_{s-2,\delta+2}(\Sigma)&\to &W^p_{s-2,\delta+2}(\Sigma) \nonumber \\
W^p_{s-2,\delta+2}(\Sigma)\times W^p_{s,\delta}(1,\Sigma)&\to &W^p_{s-2,\delta+2}(\Sigma) \nonumber 
\eea
Since each term in \eqref{4dst_ham_ctt}-\eqref{4dst_mom_ctt} can be written as a product of terms in one of $W^p_{s,\delta}$, $W^p_{s-1,\delta+1}$ or $W^p_{s-2,\delta+2}$, it follows that the constraints map into ${\cal B}_z$.

To apply the implicit function theorem, take the point $x_0\in {\cal B}_x$ to be $x_0=(0,0)$ (corresponding to vanishing EFT couplings), and $y_0=(\psi_0,Y_0)$ to be the solution of the constraints of the $\epsilon_1=\epsilon_2=0$ theory with the given free data (which exists by assumption). The Fr{\'e}chet derivative at $(x_0,y_0)$ is obtained by linearizing \eqref{4dst_ham_ctt}-\eqref{4dst_mom_ctt} in the variables $(\psi,Y)$ at $\epsilon_1=\epsilon_2=0$. Its action on the element $(0,y)$ with $y\equiv (\delta\psi,\delta Y)$ is just the linear elliptic operator
\bea
{\cal D}_y {\cal F}(0;\delta\psi,\delta Y)=\left(\begin{array}{cc}
\tilde \gamma^{kl}\tilde D_k\tilde D_l -c&\frac14\psi_0^{-7}\tilde A_{kl}(\tilde L \cdot)^{kl}  \\
0& \tilde \Delta_L 
\end{array}\right)\left(\begin{array}{c}
\delta \psi  \\
\delta Y
\end{array}\right)
\eea
Lemma \ref{poisson_lemma} tells us that the Poisson operator $\tilde \gamma^{kl}\tilde D_k\tilde D_l -c$ is an isomorphism if for every non-trivial function $f\in W^p_{2,\rho}(\Sigma)$ (with $\rho>\frac12-\tfrac{3}{p}$) the inequality \eqref{pos_func_scalar} is satisfied which is true by assumption.

Moreover, $\tilde \Delta_L $ is also an isomorphism on asymptotically Euclidean manifolds due to Lemma \ref{conf_laplace_lemma}. Since ${\cal D}_y {\cal F}$ is upper triangular and the operators in the diagonals are isomorphisms, it follows that ${\cal D}_y {\cal F}$ is also an isomorphism. Therefore, the assumptions of the implicit function theorem are satisfied which concludes the proof.
\end{proof}

We conclude this subsection by a short discussion of how the inequality \eqref{pos_func_scalar} applies to the construction of initial data for a binary system of black holes with scalar hair. We first note that if $V(\phi)\geq 0$ then the minimally coupled scalar obeys the dominant energy condition. Furthermore, on a maximal slice $R[D]=K_{ij}K^{ij}+\varrho+{4\partial}{\rm ST~ corrections}\geq 0$ which implies that the metric $\gamma$ is in the positive Yamabe class. By virtue of the conformal invariance of the functional $I_\gamma[f]$ (defined in \eqref{yamabe_inv_def}), $\tilde \gamma$ is also in the positive Yamabe class. Astrophysically relevant stationary black hole solutions in a weakly coupled $4\partial$ST theory are expected to have only a small amount of scalar hair. The reason for this is that there are no known stable scalar hairy black hole solutions in the $\epsilon_1=\epsilon_2=0$ theory. Therefore, the functional in \eqref{pos_func_scalar} is expected to be positive for scalar hairy black hole data, as well as scalar hairy black hole binary data. 

\subsection{Construction of puncture data for scalar-tensor effective field theory}\label{sec:4dst_by}

The fact that the momentum constraint has such a simple form when written in terms of the conjugate momenta has the remarkable consequence that one can straightforwardly generalize the Bowen-York-type approach to a more general class of theories. The solution to the momentum constraint with $N$ punctures located at $c_{(\alpha)}$ is given by
\be\label{BY_horndeski}
\tilde \gamma_{ij}=\delta_{ij}, \qquad \qquad
\tilde \pi_\phi=0, \qquad\qquad \tilde \pi^{ij}=B^{ij}
\ee
where the traceless Bowen-York tensor $B^{ij}$ is given by
\bea
B^{ij}&=\sum\limits_{\alpha=1}^N&\biggr[\frac{3}{2 r_{(\alpha)}^2}\Bigl(P^i_{(\alpha)}\hat x^j_{(\alpha)}+P^j_{(\alpha)}\hat x^i_{(\alpha)}-\left(\delta^{ij}-\hat x^i_{(\alpha)}\hat x^j_{(\alpha)}\right)P_k\hat x^k_{(\alpha)}\Bigr) \nonumber \\
& &+\frac{3}{r^3_{(\alpha)}}\left(\epsilon^{i}{}_{kl}S^k_{(\alpha)} {\hat x}^l_{(\alpha)} {\hat x}^j_{(\alpha)}+\epsilon^j{}_{kl}S^k_{(\alpha)} {\hat x}^l_{(\alpha)} {\hat x}^i_{(\alpha)}\right)\biggr] \label{by_tensor}
\eea
with $r_{(\alpha)}\equiv |x-c_{(\alpha)}|$ and $\hat x_{(\alpha)}\equiv \left(x-c_{(\alpha)}\right)/r_{(\alpha)}$. Note that we have not specified the initial scalar field configuration yet, \eqref{BY_horndeski} is a solution regardless of what $\phi(x)$ is.

For the moment, assume that we are given a suitably regular (specified later) initial data for $\phi$. Furthermore, as in section \ref{sec:gr_by}, it is useful to write the conformal factor as a sum of a singular piece and a regular piece $u$:
\be\label{psi_decomp}
\psi=1+\frac{1}{\mu}+u, \qquad \qquad \frac{1}{\mu}\equiv\sum\limits_{\alpha=1}^N\frac{m_{(\alpha)}}{2r_{(\alpha)}}
\ee
Then, to obtain initial data for the variables $u$, $K_\phi$ and $K^i{}_j$ (or equivalently, for $\psi$, $K$, $\tilde A_{ij}$, $\tilde K_\phi$), one needs to solve the system of equations consisting of the elliptic PDE
\bea\label{by_ham_4dst}
0&= & \partial^i\partial_i u+\frac{1}{16} \psi  \partial_i\phi\partial^i\phi-\frac18\psi^5 \left(\delta^{i_1i_2}_{j_1j_2}K_{i_1}{}^{j_1}K_{i_2}{}^{j_2}- V(\phi)-2K_\phi^2\right) \nonumber \\
& &+\frac18\epsilon_1 f_1(\phi)\psi^5\left(12 K_\phi^4-2\psi^{-4} K_\phi^2~ \partial_i\phi\partial^i\phi-\frac14 \psi^{-8}(\partial_i\phi\partial^i\phi)^2\right) \nonumber \\
& & +\frac14 \epsilon_2\psi^{-3} \biggl( \partial_{i_3}\partial^{j_3} f_2(\phi)-2\partial_{i_3}\ln\psi~\partial^{j_3} f_2(\phi)-2\partial_{i_3} f_2(\phi)\partial^{j_3}\ln\psi \nonumber \\
& &+2\delta^{j_3}_{i_3}\partial^k\ln\psi ~\partial_k  f_2(\phi)-2\psi^4  f_2^\prime(\phi) K_\phi K_{i_3}{}^{j_3}\biggr) \nonumber \\
& & \times \biggl(8\psi^{-1}\partial^{i_3}\partial_{j_3}\psi-8\psi^{-1}\delta^{i_3}_{j_3}\partial_{k}\partial^{k}\psi-24\psi^{-2}\partial^{i_3}\psi~\partial_{j_3}\psi \nonumber \\ & &-16\psi^{-2}\delta^{i_3}_{j_3}\partial_{k}\psi~\partial^{k}\psi +2\psi^4 \delta^{i_1i_2i_3}_{j_1j_2j_3}K_{i_1}^{j_1}K_{i_2}^{j_2}\biggr)
\eea
and the set of algebraic equations (that come from the combination of \eqref{BY_horndeski}, \eqref{pi_phi} and \eqref{pi_ij})
\bea
0 &= &-2K_\phi(1+2\epsilon_1 f_1(\phi) X) -\frac43 \epsilon_2 f_2^\prime (\phi)\delta^{i_1i_2i_3}_{j_1j_2j_3}K_{i_1}{}^{j_1} K_{i_2}{}^{j_2}K_{i_3}{}^{j_3} \label{st_puncture_momentum_1} \\
\psi^{-6}B^i{}_j&=&\delta^{ii_1}_{jj_1}K_{i_1}{}^{j_1}+4 \epsilon_2\delta^{ii_1i_2}_{jj_1j_2}K_{i_1}{}^{j_1}\biggl[\psi^{-4}\partial_{i_2} \partial^{j_2}  f_2(\phi)- f_2^\prime(\phi)K_{\phi} K_{i_2}{}^{j_2} \nonumber \\
& & -2\psi^{-4} \partial_k  f_2(\phi)\left(\delta_{i_2}^k \partial^{j_2} \ln \psi+\delta^{j_2k}\partial_{i_2} \ln\psi-\delta_{i_2}^{j_2}\partial^k \ln\psi\right)\biggr] \label{st_puncture_momentum_2}
\eea
In the next subsection, it will be shown that this system of equations has a unique solution for small couplings (i.e. the couplings satisfy \eqref{small} with $\Lambda$ taken to be the smallest of the bare masses) such that $u\in C^2(\mathbb{R}^3)$ and $K_\phi,K^i{}_j\in C^1(\mathbb{R}^3)$.

Several remarks are in order about this construction. First of all, note the placing of the indices in equations \eqref{st_puncture_momentum_1}-\eqref{st_puncture_momentum_2}. The main reason for writing these equations in this particular form (apart from compactness) is that $K^i{}_j$ turns out to be regular on the entire initial slice (see next section for details), whereas $K_{ij}$ is singular at the punctures.

Another natural question to ask is how to solve the system in practice. The most straightforward approach is to start with the puncture data of vacuum GR $\psi=\psi_0$, $K_\phi=0$, $K_{ij}=\psi_0^{-2}B_{ij}$ and solve the system iteratively. At weak coupling, it seems likely that such an approach would converge quickly and the required number of iteration steps is fairly small.

Now let us address the question of how to choose initial data for the scalar field. Unfortunately, there appears to be no natural choice of the initial $\phi(x)$ that could be written down in a simple closed form. Of course, one could approximate the initial value of $\phi$ with a superposition of a set of scalar cloud configurations localised around the punctures, multiplied by some regularizing function that satisfies the regularity conditions stated in the next section. But such data is likely to contain significant junk scalar radiation in addition to the junk gravitational radiation already present in the GR version of the puncture data.

As a final remark, note that the linear and angular momentum of this data is still
\be
P_{\rm ADM}^i=\frac{1}{8\pi}\int\limits_{S_\infty}\mathrm{d}A~ \pi^{ij}\hat x_j=\frac{1}{8\pi}\int\limits_{S_\infty}\mathrm{d}A~ B^{ij}\hat x_j=\sum\limits_{\alpha}P^i_{(\alpha)}
\ee
and
\be
J^i=\frac{1}{8\pi}\int\limits_{S_\infty}\mathrm{d}A~ \pi_{jk}\epsilon^{ijl} x_l\hat x^k=\frac{1}{8\pi}\int\limits_{S_\infty}\mathrm{d}A~ B_{jk}\epsilon^{ijl} x_l\hat x^k=\sum\limits_{\alpha}S^i_{(\alpha)},
\ee
respectively, as expected.

\subsection{Existence and uniqueness of puncture data for scalar-tensor effective field theory}\label{sec:4dst_by_math}

Now we demonstrate that the construction just described yields a unique initial data for the variables $\psi$ and $K_I\equiv (K_\phi,K^i{}_j)$, provided we are given $W^p_{s,\delta}(\mathbb{R}^3)$ data for $\phi$ subject to some extra requirements at the punctures. For notational convenience, let $\Pi_I\equiv ( \pi_\phi, \pi^{i}{}_j-\psi^{-6}B^{i}{}_{j})$. We proceed very similarly as in section \ref{sec:4dst_ctt_math}: we define a map $z\equiv {\cal F}(x,y)$ with
\be
x\equiv \bigl(\epsilon_1 f_1(\phi),\epsilon_2 f_2(\phi)\bigr), \qquad y\equiv (u, K_I), \qquad z\equiv \bigl({\cal H}(x,y),\Pi_I(x,y)\bigr)
\ee
with
\be
{\cal B}_x\equiv W^p_{3,\delta}(\mathbb{R}^3)\times W^p_{3,\delta}(\mathbb{R}^3), \qquad {\cal B}_y=W^p_{3,\delta}(\mathbb{R}^3)\times W^p_{2,\delta+1}(\mathbb{R}^3)
\ee
To cancel the singularities of $\psi$ and make the Einstein-scalar-field and $4\partial$ST terms regular in \eqref{by_ham_4dst}, we have to assume $\lim\limits_{r_{(\alpha)}\to 0}r_{(\alpha)}^{-1/2}\partial_i\phi=0$ and $\lim\limits_{r_{(\alpha)}\to 0}r_{(\alpha)}^{-5}V(\phi)=0$. There is no such problem with the rest of the terms in \eqref{by_ham_4dst} since they are multiplied by negative powers of the conformal factor. Now let $s\geq 3$ and $p>3$ which implies $\phi\in C^2({\mathbb{R}^3})$. Assuming $f_1(\phi),f_2(\phi)\in W^p_{s,\delta}(\mathbb{R}^3)$ and $V(\phi)\in W^p_{s-2,\delta+2}(\mathbb{R}^3)$, the Sobolev multiplication properties (Lemma \ref{w_sobolev_properties}) imply
\be
{\cal B}_z\equiv W^p_{1,\delta+2}(\mathbb{R}^3)\times W^p_{2,\delta+1}(\mathbb{R}^3)
\ee
Note that in this case ${\cal F}(x,y)=0$ is an elliptic PDE for $u$ coupled to a set of algebraic equations for $K_I$. Take $x_0=0$ (i.e. $\epsilon_1=\epsilon_2=0$) and let $y_0$ denote the original Bowen-York solution $y_0\equiv (u_0,K_{I0})\in {\cal B}_y$ described in section \ref{sec:gr_by}. We wish to use the implicit function theorem to show that for small enough couplings, the equations \eqref{by_ham_4dst}-\eqref{st_puncture_momentum_1} admit a unique solution. The derivative of ${\cal F}$ at $(x_0,y_0)$ acting on the element $(0,y)$ with $y\equiv (\delta u,\delta K_I)$ is given by
\bea
{\cal D}_y {\cal F}\bigr|_{(x_0,y_0)}(0,y)\equiv \left(\begin{array}{c}
\delta {\cal H}  \\
\delta \pi_\phi \\
\delta \pi^i{}_{j}
\end{array}\right)=\left(\begin{array}{ccc}
\tilde \gamma^{kl}\tilde D_k\tilde D_l -\frac78 \psi_0^{-8}B_{kl}B^{kl}& 0 &\frac14 \psi_0^{-1} B^l{}_{k}  \\
0&2 &0  \\
0 & 0 & \delta^{il}_{jk}
\end{array}\right)\left(\begin{array}{c}
\delta u  \\
\delta  K_\phi \\
\delta K^k{}_l
\end{array}\right)
\eea
which is clearly invertible (due to Lemma \ref{poisson_lemma}) and thus the map $\cal F$ obeys the conditions of Theorem \eqref{thm:imp}. Therefore, we have the following result.

\begin{theorem}[Existence and uniqueness of puncture data]
Let $\phi\in W^p_{s,\delta}(\mathbb{R}^3)$ such that $s\geq 3$, $p>3$ and $0\leq \delta <1-\frac{3}{p}$. Then the system \eqref{by_ham_4dst}-\eqref{st_puncture_momentum_2} admits a unique solution in the neighbourhood of the original puncture solution of section \ref{sec:gr_by} for small enough $\epsilon_1$, $\epsilon_2$. The solution for the conformal factor $\psi$ is of the form \eqref{psi_decomp} with $u\in W^p_{3,\delta}(\mathbb{R}^3)$ and $K_\phi, K^i{}_j\in W^p_{2,\delta+1}(\mathbb{R}^3)$, provided that
\begin{itemize}
    \item[(i)] $\lim\limits_{r_{(\alpha)}\to 0}r_{(\alpha)}^{-1/2}~\partial_i\phi=0$ and $\lim\limits_{r_{(\alpha)}\to 0}r_{(\alpha)}^{-5}~V(\phi)=0$.
    \item[(ii)] $f_1(\phi)-f_1(\phi_\infty),f_2(\phi)-f_2(\phi_\infty)\in W^p_{3,\delta}(\mathbb{R}^3)$ and $V(\phi)\in W^p_{1,\delta+2}(\mathbb{R}^3)$
\end{itemize} 
Moreover, $u\in C^2(\mathbb{R}^3)$ and $K_\phi, K^i{}_j\in C^1(\mathbb{R}^3)$.
\end{theorem}

\section{The conformal thin sandwich method}\label{Sec:CTS}

\subsection{The conformal thin sandwich equations}

An alternative but closely related way to construct asymptotically Euclidean initial data is the so-called conformal thin sandwich (CTS) method, originally proposed by York \cite{York:1998hy}. The mathematical formulation of the method only slightly differs from the CTT approach but the difference is important from a physics point of view: the CTS version provides a more natural way to prepare initial data that is nearly stationary. This is an improvement compared to the Bowen-York data in the sense that one can significantly reduce the amount of junk radiation plaguing the initial stages of a numerical simulation. However, there is a cost to pay for this: to achieve such data, one has to solve the momentum constraint, as well as the Hamiltonian constraint using an elliptic solver.

Once again, our starting point is the conformal metric \eqref{conf_metric} and the decomposition of the extrinsic curvature \eqref{extr_K_trace}. The main difference compared to the CTT approach is the further treatment of the variable $\tilde A_{ij}$. Let $\alpha$ denote the lapse function and $\beta^i$ the shift vector. Then introducing\footnote{In some numerical relativity applications the conformal metric is chosen such that $\det \tilde\gamma=1$ and in that case $U^{ij}=\partial_t \tilde \gamma^{ij}$.}
\be\label{def_U}
U^{ij}\equiv (\det\tilde\gamma)^{-1/3}\partial_t \left[\tilde \gamma^{ij}(\det\tilde\gamma)^{1/3}\right] \qquad \qquad {\rm and} \qquad \qquad \tilde \alpha \equiv \psi^{-6}\alpha,
\ee
allows us to write
\be\label{A_cts}
\tilde A^{ij}=\frac{1}{2\tilde \alpha} \left[(\tilde L \beta)^{ij}+U^{ij}\right]
\ee
where the conformal Killing operator $\tilde L$ was defined in \eqref{conf_KVF}. Thus in the CTS decomposition the variables $U^{ij}$ and $\beta^i$ are used instead of the CTT variables $\tilde A_{\rm TT}^{ij}$ and $Y^i$. Plugging \eqref{A_cts} into the Hamiltonian and momentum constraints \eqref{gr_ham}-\eqref{gr_mom} of General Relativity gives elliptic equations for the variables $(\psi,\beta^i)$, whereas the 'free data' now consists of the sextuple $(\tilde \alpha,\tilde \gamma_{ij}, U^{ij},K,\tilde \varrho, \tilde{\mathfrak{p}}^i)$. Therefore, the CTS constraints can be written as
\bea
 \tilde \gamma^{kl}\tilde D_k\tilde D_l \psi-\frac18 \psi  R[\tilde D]-\frac{1}{12}\psi^5 K^2+\frac18\psi^{-7}\tilde A_{kl}\tilde A^{kl}+\frac18 \psi^{-3}\tilde \varrho&=0 \\
  \tilde D_j \tilde A^{ij}-\frac23 \psi^6 \tilde \gamma^{ij}\tilde D_j K-\frac12 \tilde{\mathfrak{p}}^i&=0.
\eea
where $\tilde A$ is to be replaced by the RHS of \eqref{A_cts}. As mentioned, this system has very similar properties as the CTT system. In particular, given suitable free data, a unique solution exists to the corresponding elliptic boundary value problem under very similar conditions. Moreover, the Hamiltonian constraint decouples from the momentum constraint on CMC slices as in the case of the CTT system.

Consider now the scalar-tensor theory \eqref{4dST} at weak coupling. The CTS equations for this theory can be obtained in exactly the same way: i.e., by substituting the decomposition \eqref{A_cts} into \eqref{4dst_ham_ctt}-\eqref{4dst_mom_ctt}. For small enough $\epsilon_1$ and $\epsilon_2$, this is an elliptic PDE system for $(\psi,\beta^i)$. In the next section, we state a well-posedness theorem for the corresponding elliptic boundary value problem.

\subsection{Mathematical statements}

The similarity of the CTT and CTS equations allows us to straightforwardly transform statements about one of these systems to statements about the other one. The only extra requirement we need concerns the conformal lapse function $\tilde \alpha$. It is customary to choose the lapse so that its asymptotic value at spatial infinity approaches 1. Hence, the natural assumption on $\tilde \alpha$ (which is part of the free data) is $\tilde\alpha\in W^p_{s,\delta}(1,\Sigma)$ (defined in \eqref{def:sob_1}) and $\tilde \alpha>0$. Then the methods used to obtain Theorem \ref{thm_esf_constraints} imply the following.
\begin{theorem}\label{thm_esf_cts_constraints}
Let $(\Sigma,\gamma)$ be a $3$-dimensional asymptotically Euclidean manifold of class $W^p_{s,\delta}$ with $K=0$, $p>\tfrac32$, $s>2+\tfrac{3}{p}$ and $1-\tfrac{3}{p}>\delta>-\tfrac{3}{p}$. Moreover, suppose
$$R[\tilde D]-\frac12\tilde \gamma^{ij}\partial_i\phi\partial_j\phi>0.$$
Then there exists an open set of values for the free data $(\tilde \alpha,U,\phi,\tilde K_\phi)$ satisfying $\tilde\alpha\in W^p_{s,\delta}(1,\Sigma)$, $\tilde \alpha>0$, $\phi-\phi_\infty\in W^p_{s,\delta}(\Sigma)$, $U,\tilde K_\phi \in W^p_{s-1,\delta+1}(\Sigma)$ and $V(\phi)\in W^p_{s-2,\delta+2}(\Sigma)$ such that the conformally formulated constraints have a solution $(\psi,\beta^i)$ with $\psi\in W^p_{s,\delta}(1,\Sigma)$, $\psi>0$ and $\beta\in W^p_{s,\delta}(\Sigma)$.
\end{theorem}
Likewise, we have the CTS analogue of Theorem \ref{thm:4dst_ctt} for the $4$-derivative scalar tensor theory \eqref{4dST}. To avoid repetition, we only state this result without proof. (The proof is completely analogous to that of Theorem \ref{thm:4dst_ctt}: it is a combination of the weighted Sobolev multiplication lemma and the implicit function theorem.)
\begin{theorem}\label{thm:4dst_cts}
Let $(\Sigma,\tilde \gamma)$ be a $3$-dimensional asymptotically Euclidean manifold of type $W^p_{s,\delta}$ with
$$p>\frac{3}{2},\qquad s>3+\frac{3}{p} \qquad {\rm and} \qquad \frac{1}{2}-\frac{3}{p}<\delta<1-\frac{3}{p}.$$
Assume that we are given a solution of the CTS Einstein-scalar-field constraint equations $(\psi_0,\beta_0^i)$ with $\psi_0>0$ and $\psi_0-1,\beta_0^i \in W^p_{s,\delta}(\Sigma)$ corresponding to free data $(\tilde\alpha,\tilde \gamma,K,U,\phi,\tilde K_\phi)$ satisfying $\tilde\alpha\in W^p_{s,\delta}(1,\Sigma)$, $\tilde \alpha>0$, $\phi-\phi_\infty \in W^p_{s,\delta}(\Sigma)$, $V(\phi)\in W^p_{s-2,\delta+2}(\Sigma)$, $\tilde U, \tilde K_\phi \in W^p_{s-1,\delta+1}(\Sigma)$ and $K=0$.

Then the $4\partial$ST constraint system admits a unique solution $(\psi,\beta^i)$ for small enough $\epsilon_1,\epsilon_2$, provided that $f_1(\phi)-f_1(\phi_\infty),f_2(\phi)-f_2(\phi_\infty)\in W^p_{s,\delta}(\Sigma)$ for every non-trivial function $f\in W^p_{2,\rho}(\Sigma)$ with $\rho>\frac12-\tfrac{3}{p}$ the following inequality holds:
    \be\label{pos_func_scalar_2}
    \int\limits_\Sigma d^n x\sqrt{\tilde\gamma} \left[\tilde \gamma^{ij}\partial_i f\partial_j f+cf^2\right]>0
    \ee
    with
    \be
    c\equiv \frac18\left( R[\tilde D]+7\psi_0^{-8}\tilde A_{kl}\tilde A^{kl}+14 \psi_0^{-6}\tilde  K_\phi^2-\frac12 \tilde \gamma^{ij}\partial_i\phi\partial_j\phi-5\psi_0^4 V(\phi)\right) \nonumber
    \ee
The $4\partial$ST solution $(\psi,\beta^i)$ is close to $(\psi_0,\beta^i_0)$ in the sense of $W^p_{s,\delta}(\Sigma)$ norms.
\end{theorem}
Similar results hold in the case when $K$ is not identically but nearly zero and of class $W^p_{s-1,\delta+1}$ (which can be proved by using the implicit function theorem).

\subsection{Choice of free data for black hole binaries}\label{sec:cts_sup}

We end the discussion of the original conformal thin sandwich method by giving a brief account of how free data can be chosen for numerical simulations in General Relativity and $4\partial$ST theories that approximates binary black hole systems. 

In General Relativity, a standard choice for the free data (for black hole binaries) is to take a "superposition" of two isolated Kerr solutions, with their respective coordinate origins shifted appropriately (see e.g. \cite{Lovelace_2009}). The idea is that as long as the separation of the two black holes is sufficiently large, the superposed geometry is a good first approximation of the actual geometry of the binary system. Then solving the constraints for $(\psi,\beta^i)$ accounts for the deviation of this naive data from the actual binary black hole geometry. In more detail, let $\gamma_{ij}^{(1)}$, $\alpha^{(1)}$, $\beta^{(1)i}$, $\dot \gamma_{ij}^{(1)}$ be the induced metric, the lapse function, the shift vector and the time derivative of the induced metric, respectively, of the first Kerr black hole, written in a suitable coordinate system (usually Kerr-Schild coordinates). We use a similar notation for the second black hole. One can take the naive conformal metric and conformal lapse to be
\bea
\tilde \gamma_{ij}&=&\gamma_{ij}^{(1)}+ \gamma_{ij}^{(2)}-\delta_{ij} \nonumber \\
\tilde\alpha &=&\alpha^{(1)}+\alpha^{(2)}-1 \nonumber
\eea
and introduce an auxiliary shift vector defined by
\be
\beta^i_{\rm aux}=\beta^{(1)i}+\beta^{(2)i} \nonumber
\ee
Inspired by the identities
\bea
U^{ij}&=&\psi^4\left(\partial_t \gamma^{ij}-\frac13 \gamma^{ij}\gamma_{kl}\partial_t \gamma^{kl}\right) \nonumber \\
K&=&\frac{1}{2\alpha}\left(\gamma_{ij}\partial_t\gamma^{ij}+2\partial_i\beta^i+\gamma^{ij}\beta^k\partial_k\gamma_{ij}\right) \nonumber
\eea
and taking $\psi_{aux}=1$ as the naive approximate solution for the conformal factor, we can choose the rest of the free data as
\bea
U^{ij}&=&-\tilde\gamma^{ik}\tilde\gamma^{jl}\left(U^{\rm aux}_{kl}-\frac13 \tilde \gamma_{kl} \tilde \gamma^{mn}U^{\rm aux}_{mn}\right) \qquad \qquad {\rm with} \qquad U^{\rm aux}_{ij}\equiv \dot \gamma_{ij}^{(1)}+\dot \gamma_{ij}^{(2)} \nonumber \\
K&=&\frac{1}{2\tilde \alpha}\left(\tilde \gamma_{ij}U^{ij}+2\partial_i\beta^i_{\rm aux}+\tilde \gamma^{ij}\beta^k_{\rm aux}\partial_k\tilde \gamma_{ij}\right) \nonumber
\eea

The construction just described works essentially the same way for the scalar-tensor theories \eqref{4dST}. The only difference is that for a $4\partial$ST theory the gravitational part of the free data is constructed using a black hole solution of the theory which for generic couplings differs from the Kerr solution. The natural choice for the scalar part of the free data is to follow the philosophy of "superposition" and take
\be
\phi=\phi^{(1)}+\phi^{(2)}-\phi_\infty, \qquad \qquad \dot\phi=\dot\phi^{(1)}+\dot\phi^{(2)}
\ee
where $\phi^{(1)}$, $\phi^{(2)}$ and $\dot\phi^{(1)}$, $\dot\phi^{(2)}$ are the scalar field configurations of the black holes and their time derivatives, $\phi_\infty$ is the asymptotic value of the scalar field at spatial infinity (which may be non-zero).

\section{The extended conformal thin sandwich method}\label{sec:xcts}

\subsection{Quasiequilibrium initial data for General Relativity}\label{sec:quasieq}

A proposal for initial data in GR representing a binary black hole system in quasiequilibrium is provided by an extension of the Conformal Thin Sandwich method \cite{Pfeiffer:2002iy}.

The purpose of this approach is to construct initial data for two black holes moving along circular orbits. In other words, the binary system is assumed to be in quasiequilibrium, meaning that the spacetime possesses an approximate helical Killing vector field $\xi_{\rm hel}$ that generates circular orbits. Given such a vector field, one could choose coordinates which co-rotate with the binary system. This amounts to choosing a time coordinate such that $\partial/\partial t$ is parallel with $\xi_{\rm hel}$. (Of course, with this choice $\partial/\partial t$ will not be timelike everywhere.) Then the requirement of quasiequilibrium can be expressed as $\partial_t g_{\mu\nu}\approx 0$. Using the conformal variables, one could set
\be\label{quasieq}
\partial_t \tilde \gamma_{ij}=0, \qquad \qquad \partial_t K=0
\ee
on the initial slice. The first of these two criteria implies
\be\label{quasieq_A}
\tilde A^{ij}=\frac{\psi^6}{2\alpha} (\tilde L \beta)^{ij}
\ee
Substituting this into the momentum constraint yields
\be\label{cts_shift}
\tilde \Delta_L\beta^i-(\tilde L \beta)^{ij}\tilde D_j\ln(\psi^{-6}\alpha)=\frac43\alpha \tilde D^i K
\ee
In the extended CTS approach one obtains an extra elliptic equation that involves the lapse function. To derive this equation, consider the following combination of the equations of motion
\be\label{xcts_lin_comb}
\frac14\alpha \psi^5\left(2\gamma^{\mu\nu}G_{\mu\nu}+{\cal H}\right)=0
\ee
The only second time derivative term in equation \eqref{xcts_lin_comb} is $\partial_t K$ so using \eqref{quasieq} gives the constraint equation
\be\label{cts_lapse}
\tilde \gamma^{ij}\tilde D_i\tilde D_j(\alpha\psi)=\alpha\psi\left(\frac78\psi^{-8}\tilde A_{ij}\tilde A^{ij}+\frac{5}{12}\psi^4 K^2+\frac18 R[\tilde D]\right)+\psi^5\beta^i\tilde D_i K
\ee
Therefore, the system of elliptic equations to be solved for $\psi$, $\alpha\psi$ and $\beta^i$ consists of \eqref{cts_shift}, \eqref{cts_lapse} and the Hamiltonian constraint \eqref{Ham_GR_v2} (with $\tilde A^{ij}$ set to be equal to the RHS of \eqref{quasieq_A}).

Having written down the elliptic equations, it remains to specify suitable boundary conditions. Following \cite{Cook:2004kt}, one can solve the system in a region with a single asymptotically flat end and a finite number of interior boundaries, corresponding to the apparent horizons of the black holes. The boundary conditions at spatial infinity are
\be\label{xcts_bc_1}
\lim\limits_{r\to \infty}\psi =1, \qquad \lim\limits_{r\to \infty}\alpha =1, \qquad \lim\limits_{r\to \infty}\beta^i =\Omega_{\rm orb}\epsilon^{ijk}e_jx_k 
\ee
where $\Omega_{\rm orb}$ is the orbital angular velocity of the system, as measured at infinity and $e^i$ is the unit vector specifying the axis of rotation. One may worry that with these boundary condition the shift vector diverges at spatial infinity. However, this is just an artefact of the corotating coordinates and the solution for $\beta^i$ can be written as
$$\beta^i =\Omega_{\rm orb}\epsilon^{ijk}e_jx_k+\beta^i_{\rm reg} $$
where $\beta^i_{\rm reg}$ is regular and has the usual asymptotic fall-off. To make sure that the numerical evolution of this data is stable, one can simultaneously work in two different coordinate system \cite{Scheel:2006gg}: one of them is the corotating coordinate system just described, the other one is "inertial".

To fix the boundary conditions on the interior boundaries, one can proceed as in \cite{Cook:2004kt} to require that the interior boundary surfaces correspond to apparent horizons. Let $S$ be one of these 2-surfaces and let $s^a$ be the unit normal (w.r.t. $\gamma$) to $S$. Then one can define the induced metric on $S$
\be
P_{ij}=\gamma_{ij}-s_is_j
\ee
It is useful to introduce the conformally rescaled version of $P_{ij}$ and $s^i$
\be
P_{ij}=\psi^4\tilde P_{ij}, \qquad \qquad s_i=\psi^2\tilde s_i
\ee
Furthermore, let $\tilde{\cal D}$ denote the covariant derivative associated with $\tilde P$. Then one can impose the following requirements and a set of corresponding boundary conditions \cite{Cook:2004kt}.
\begin{itemize}
    \item[(1)] First of all, one can construct a future-directed null vector field
    $$k^\mu=\frac{1}{\sqrt2}(n^\mu+s^\mu) $$
    and consider the null geodesic congruence with tangent vectors $k^\mu$ that pass through the surface $S$. These null geodesics determine a null hypersurface in the neighborhood of $S$. A possible notion of quasiequilibrium is to require that coordinate system initially tracks this null hypersurface, i.e. the coordinate location of the surface $S$ does not change initially.
    Then this condition translates to $\left(\frac{\partial}{\partial t}\right)^\mu k_\mu\bigr|_S=0$ which implies
    \be\label{xcts_b2}
    \left(\beta_\perp-\alpha\right) \bigr|_{S}=0, \qquad \qquad \beta^i_\perp\equiv \beta^is_i
    \ee
    \item[(2)] The condition that the expansion of the null geodesic congruence defined above must be zero yields a Neumann-type boundary condition on the conformal factor
    \be\label{xcts_b3}
    \left[\tilde s^k\tilde D_k\ln\psi-\frac14\left(\psi^2P^{ij}K_{ij}-\tilde P^{ij}\tilde D_i\tilde s_j\right)\right]\biggr|_{S}=0
    \ee
    \item[(3)] The vanishing of the shear of the null geodesic congruence together with $\partial_t \tilde \gamma_{ij}=0$ gives an equation for $\beta^i\equiv P^i{}_j\beta^j$
    \be\label{bc_ckv}
    \left(\tilde{\cal D}^{(i}\beta^{j)}_{\parallel}-\frac12 \tilde P^{ij}\tilde{\cal D}_k\beta^k_{\parallel}\right)\biggr|_{S}=0.
    \ee
    Interestingly, this is just the conformal Killing equation on $S$. Since any closed 2d surface is conformally equivalent to the unit 2-sphere, the problem of solving \eqref{bc_ckv} boils down to finding Killing vector fields on the unit 2-sphere. The unit 2-sphere has a family of rotational Killing vector fields $\xi^i$. Hence the vector
    \be\label{xcts_b4}
    \beta^i_\parallel=\Omega_{\rm s} \xi^i
    \ee
    solves equation \eqref{bc_ckv} for any constant parameter $\Omega_s$. The parameter $\Omega_s$ determines the magnitude of the angular velocity of the black hole whereas $\xi_i$ determines the axis of rotation.
    \item[(4)] The boundary value of the lapse function $\alpha$ can be chosen freely as part of the initial temporal gauge freedom.
\end{itemize}

To summarize, the construction of initial data in quasiequilibrium amounts to solving the coupled system of elliptic equations \eqref{Ham_GR_v2}, \eqref{cts_shift} and \eqref{cts_lapse} subject to the boundary conditions \eqref{xcts_bc_1}, \eqref{xcts_b2}, \eqref{xcts_b3}, \eqref{xcts_b4}. The initial values of $\tilde \gamma_{ij}$, $K$ are source terms in these equations and can be chosen freely, e.g. as $\tilde \gamma_{ij}=\delta_{ij}$ and $K=0$ or as a "superposition" of two black hole solutions, see section \ref{sec:cts_sup}.

\subsection{The Conformal Thin Sandwich equations for scalar-tensor effective field theories}\label{sec:xcts_4dst}

In this section we propose a way to adapt the extended CTS method for $4\partial$ST theories. For a scalar-tensor theory, we would like to derive elliptic equations for not only the variables $\psi$, $\alpha$ and $\beta^i$ but also for the initial value of $\phi$. Then the free part of the data consists of the initial values of $(\tilde \gamma, \partial_t \tilde \gamma, K, \partial_t K, K_\phi, \partial_t K_\phi)$.

The Hamiltonian and momentum constraints can be converted to elliptic equations in the exact same way as in GR: we can just take equations \eqref{4dst_ham_ctt}, \eqref{4dst_mom_ctt} and substitute $\tilde A^{ij}$ with $\frac{\psi^6}{2\alpha} (\tilde L \beta)^{ij}$ (c.f. equation \eqref{quasieq_A}) to impose quasiequilibrium. 

To obtain the analogue of \eqref{cts_lapse} and an elliptic equation for the scalar field, we start by considering the linear combination
$$\frac14\alpha \psi^5\left(2\gamma^{\mu\nu}E_{\mu\nu}+{\cal H}\right)=0 $$
of the $4\partial$ST theory, write it in terms of the variables $\alpha, \beta, \tilde\gamma, \psi, K, \tilde A$ and impose a suitable quasiequilibrium condition. Introducing
\bea
{\cal A}_{kl}&=&\frac{1}{\alpha}(\partial_t-\mathcal{L}_\beta)K_{kl}-K_{km}K^m{}_l-\frac{1}{\alpha}D_kD_l\alpha \\
{\cal A}_\phi&=&\frac{1}{\alpha}(\partial_t-\mathcal{L}_\beta)K_\phi+\frac{1}{2\alpha}D^k\alpha D_k\phi
\eea
and using the equations of Appendix \ref{app:adm}, this combination yields
\bea\label{4dst_xcts_lapse}
0&=&\tilde \gamma^{ij}\tilde D_i\tilde D_j(\alpha\psi)-\alpha\psi\left(\frac78\psi^{-8}\tilde A_{ij}\tilde A^{ij}+\frac{5}{12}\psi^4 K^2+\frac18 R[\tilde D]\right)+\psi^5(\partial_t-\beta^i\tilde D_i) K \nonumber \\
& & +\frac14 \alpha \psi^5\biggl\{ 16\gamma^{i[i_1}\gamma^{j]j_1}\left(D_{i_1}D_{j_1}\epsilon_2 f_2(\phi)-2\epsilon_2 f_2^\prime(\phi)K_\phi K_{i_1j_1}\right){\cal A}_{ij}\nonumber \\
& &+\frac14(1-3\epsilon_1 f_1(\phi) X)D^i\phi D_i\phi+\frac52 V(\phi)-K_\phi^2\left(7+9\epsilon_1 f_1(\phi) X\right) \nonumber \\
& &-\epsilon_2\gamma^{i_1i_2i_3}_{j_1j_2j_3}\left(R[D]_{i_1i_2}{}^{j_1j_2}+2K_{i_1}{}^{j_1}K_{i_2}{}^{j_2}\right)\left(D_{i_3}D^{j_3}f_2(\phi)-2 f_2^\prime(\phi)K_\phi K_{i_3}{}^{j_3}\right) \nonumber \\
& &+4\epsilon_2\gamma^{i_1i_2}_{j_1j_2}\biggl[\left(R[D]_{i_1i_2}{}^{j_1j_2}+2 K_{i_1}{}^{j_1}K_{i_2}{}^{j_2}\right)\left( f_2^\prime(\phi){\cal A}_\phi-2  f_2^{\prime\prime}(\phi)K_\phi^2\right) \nonumber \\
& & -2\epsilon_2\left(\gamma_{i_1}^{j_3}\gamma_{i_3}^{j_1}+\gamma^{j_1j_3}\gamma_{i_1i_3}\right)\left(-K_{j_3}{}^k D_k f_2(\phi)+2D_{j_3}( f_2^\prime(\phi)K_\phi )\right)D^{i_3}K_{i_2}{}^{j_2}\biggr]\biggr\}
\eea
where the terms in the second to sixth lines are to be converted to the conformal variables using the equations of Appendix \ref{app:conf}. Equation \eqref{4dst_xcts_lapse} contains second time derivatives via the terms ${\cal A}_\phi$ and ${\cal A}_{ij}$. The natural way to impose quasiequilibrium on the data is to choose
\be
\partial_t \tilde \gamma_{ij}=0, \qquad \qquad K_\phi=0
\ee
and set the combination of time derivatives of $K_{ij}$ and $K_\phi$ appearing in \eqref{4dst_xcts_lapse} to zero. Using these conditions in \eqref{4dst_xcts_lapse} and the scalar equation of motion yields an elliptic equation for the initial values of $\alpha\psi$ and $\phi$:
\bea
0&=&\tilde \gamma^{ij}\tilde D_i\tilde D_j(\alpha\psi)-\alpha\psi\left(\frac78\psi^{-8}\tilde A_{ij}\tilde A^{ij}+\frac{5}{12}\psi^4 K^2+\frac18 R[\tilde D]\right)-\psi^5\beta^i\tilde D_i K \nonumber \\
& & +\frac14 \alpha \psi^5\biggl\{ -16\epsilon_2\gamma^{i[i_1}\gamma^{j]j_1}D_{i_1}D_{j_1} f_2(\phi)\left(K_{ik}K^k{}_j+\frac{1}{\alpha}D_iD_j\alpha\right)\nonumber \\
& & \qquad\qquad +\frac14(1-3\epsilon_1 f_1(\phi) X)D^i\phi D_i\phi+\frac52 V(\phi) \nonumber \\
& & \qquad\qquad -\epsilon_2\gamma^{i_1i_2i_3}_{j_1j_2j_3}\left(R[D]_{i_1i_2}{}^{j_1j_2}+2K_{i_1}{}^{j_1}K_{i_2}{}^{j_2}\right)D_{i_3}D^{j_3} f_2(\phi) \nonumber \\
& & \qquad\qquad +4\epsilon_2\gamma^{i_1i_2}_{j_1j_2}\biggl[ f_2^\prime(\phi)\frac{1}{2\alpha}D^k\alpha D_k\phi\left(R[D]_{i_1i_2}{}^{j_1j_2}+2 K_{i_1}{}^{j_1}K_{i_2}{}^{j_2}\right) \nonumber \\
& & \qquad\qquad +2\left(\gamma_{i_1}^{j_3}\gamma_{i_3}^{j_1}+\gamma^{j_1j_3}\gamma_{i_1i_3}\right)K_{j_3}{}^k D_k f_2(\phi)D^{i_3}K_{i_2}{}^{j_2}\biggr]\biggr\} \label{4dst_alpha_xcts}\\
0&=&(1+6\epsilon_1 f_1(\phi)X)\left(\tilde D^k \tilde D_k\phi+2\tilde D^k \phi \tilde D_k\ln\psi\right)-\psi^4 V'(\phi)+\left(1-\epsilon_1 f_1(\phi)(D\phi)^2\right)\tilde D^i\phi \tilde D_i\ln\alpha  \nonumber \\ 
& &+\psi^4\biggl\{ -\frac34 \epsilon_1 f_1^\prime(\phi)(D^k\phi D_k\phi)^2+2\epsilon_1 f_1(\phi)\gamma^{i_1i_2}_{j_1j_2}D_{i_1}\phi D^{j_1}\phi D_{i_2}D^{j_2}\phi+8\epsilon_2 f_2^\prime(\phi)\gamma^{i_1i_2i_3}_{j_1j_2j_3}D^{j_1}K_{i_2}{}^{j_2}D_{i_1} K_{i_3}{}^{j_3} \nonumber \\
& & -2\epsilon_2 f_2^\prime(\phi)\gamma^{ii_1i_2}_{jj_1j_2}\left(R[D]_{i_1i_2}{}^{j_1j_2}+2K_{i_1}{}^{j_1}K_{i_2}{}^{j_2}\right)\left(\frac{1}{\alpha}D_iD^j\alpha+K_{ik}K^{jk}\right)\biggr\} \label{4dst_phi_xcts}
\eea
Of course, again, these equations need to be rewritten in terms of the conformal variables. Since this procedure is straightforward (using the identities of Appendix \ref{app:conf}) but not particularly enlightening, we spare the reader from detailing these equations in their full lengths.

The final step in the construction is to impose boundary conditions. The apparent horizon and asymptotic boundary conditions on the gravitational variables can be taken to be the same as in section \ref{sec:quasieq} (equations \eqref{xcts_bc_1}, \eqref{xcts_b2}, \eqref{xcts_b3}, \eqref{xcts_b4}). Regarding the scalar field, it is simplest to choose
\be\label{bc_phi_xcts}
\lim\limits_{r\to \infty}\phi=\phi_\infty \qquad \qquad {\rm and } \qquad \qquad s^i\partial_i\phi\bigr|_S=0
\ee
where $s^i$ is the unit normal to the interior boundary $S$ and $\phi_\infty$ is a constant. For a single stationary black hole, the latter condition captures the property that there is no scalar flux through the horizon.

To conclude our proposal, the extended conformal thin sandwich method for the scalar-tensor theory \eqref{4dST} amounts to solving the elliptic equations \eqref{4dst_ham_ctt}, \eqref{4dst_mom_ctt}, \eqref{4dst_alpha_xcts} and \eqref{4dst_phi_xcts} for the variables $(\psi,\beta^i, \alpha, \phi)$ subject to the boundary conditions \eqref{xcts_bc_1}, \eqref{xcts_b2}, \eqref{xcts_b3}, \eqref{xcts_b4} and \eqref{bc_phi_xcts}. The (remaining) free part of the data $(\tilde \gamma, K)$ can be fixed using the methods of section \ref{sec:cts_sup}.

\subsection{On the existence and uniqueness of the extended Conformal Thin Sandwich system}\label{sec:xcts_discussion}

The problem of existence and uniqueness of the extended conformal thin sandwich system is more complicated than in the case of the CTT or the original CTS system, even in General Relativity. In particular, the issue is with the extra equation \eqref{cts_lapse}. Concentrate on the terms in \eqref{cts_lapse}
$$\tilde\gamma^{ij}\tilde D_i\tilde D_j(\alpha\psi)-\alpha\psi\left(\frac78\psi^{-8}\tilde A_{ij}\tilde A^{ij}\right)=\tilde\gamma^{ij}\tilde D_i\tilde D_j(\alpha\psi)-\frac{7}{32}\psi^6(\tilde L \beta)^{ij}(\tilde L \beta)_{ij}(\alpha\psi)^{-1}$$
Upon linearization of \eqref{cts_lapse}, this part of the equation gives rise to the terms
$$\tilde\gamma^{ij}\tilde D_i\tilde D_j\delta (\alpha\psi)+\frac{7}{32}(\alpha\psi)^{-2}\psi^6(\tilde L \beta)^{ij}(\tilde L \beta)_{ij}\delta(\alpha\psi)+...=0 $$
However, for a linear elliptic equation of the form
$$\tilde\gamma^{ij}\tilde D_i\tilde D_j\Phi-c\Phi=0 $$
with a function $c$, the uniqueness of the solution may fail if $c$ is negative\footnote{The solution is unique on asymptotically Euclidean manifolds when $c\geq 0$, c.f. Lemma \ref{poisson_lemma}.}. Therefore, one could find solutions so that $(\tilde L \beta)^{ij}(\tilde L \beta)_{ij}$ is large enough and the linearized version of \eqref{cts_lapse} does not have a unique solution.

Going beyond this toy argument, it has been explicitly demonstrated  \cite{Walsh:2006au,Pfeiffer:2005jf,Baumgarte:2006ug,Holst:2012vm} (both numerically and by means of bifurcation theory) that there exists choices of free data for the extended CTS system such that the corresponding solution is not unique. Nevertheless, the extended CTS approach has been used in several numerical simulations to construct binary black hole initial data (see e.g. \cite{num_rel} and the references therein).

Similar issues may be expected in the scalar-tensor version of the extended CTS method. In addition to this, the scalar elliptic equation \eqref{4dst_phi_xcts} may also suffer from this type of failure of uniqeness under some choice of the potential $V(\phi)$ and the couplings $f_1(\phi)$, $f_2(\phi)$. In fact, there is already some numerical evidence that this may happen. Several recent studies focused on stationary, axisymmetric black hole solutions of the theory $f_1(\phi)=V(\phi)=0$, $f_2(\phi)=\eta \phi^2$ where $\eta$ is a constant (see e.g. \cite{Herdeiro:2020wei,Doneva:2017bvd,Cunha:2019dwb,Dima:2020yac,Berti:2020kgk}). These studies concluded that if the coupling constant $\eta$ is large enough (i.e. in the strongly coupled regime where $\eta$ has order of magnitude $M^2$ where $M$ is the mass of the black hole) then the Kerr spacetime is no longer the unique rotating black hole solution: for fixed mass and angular momentum, there exists another solution with a nontrivial scalar configuration. Moreover, the scalarized solution appears to be stable. More generally, $\phi=0$ is a solution to $E_\phi=0$ and the theory inherits solutions of vacuum General Relativity. Consider the simple case when the initial slice has a single interior boundary and the free data $(\tilde\gamma,K)$ is chosen to be the induced metric and the trace of the extrinsic curvature of a slice of the Kerr spacetime (e.g. a Kerr-Schild slice). Clearly, for this particular theory, $\phi=0$ is a solution to \eqref{4dst_phi_xcts} (provided that $\phi_\infty=0$). Then the extended CTS equations are the same as in vacuum GR and they reproduce a slice of the Kerr solution \cite{Cook:2000vr}. However, for large enough $\eta$, \eqref{4dst_phi_xcts} is likely to have a bifurcating branch of solutions: for a choice of free data that represents a slice of a Kerr black hole, there will be a solution with a non-trivial scalar field configuration in addition to the solution with $\phi\equiv 0$. There is no reason to expect that the solution with non-trivial $\phi$ is an exact slice of the stationary scalar hairy black hole solution found in \cite{Herdeiro:2020wei,Doneva:2017bvd,Cunha:2019dwb,Dima:2020yac,Berti:2020kgk}. Instead, this data is more likely to represent a slice of a dynamical black hole that may eventually settle down to the stationary scalar hairy black hole of \cite{Herdeiro:2020wei,Doneva:2017bvd,Cunha:2019dwb,Dima:2020yac,Berti:2020kgk} when evolved in time. To obtain data that is a more accurate representation of a slice of a scalarized stationary black hole, one could choose the free data $(\tilde\gamma,K)$ to be the induced metric and the trace of the extrinsic curvature of a constant time slice of the scalarized black hole solution. The construction of initial data for a binary system of such black holes could then be done by using the idea of "superposition" (section \ref{sec:cts_sup}) to choose the free data.

\section{Discussion and Summary}\label{sec:discussion}

In this paper, we considered the initial data problem on asymptotically Euclidean hypersurfaces in a class of (weakly coupled) scalar-tensor effective theories, using standard methods known from the elliptic PDE literature. These methods seem to be quite robust and are applicable in more general settings. For example, we did not make use of the specific form of the equations of motion so it seems likely that the construction of asymptotically Euclidean initial data works similarly for any weakly coupled Horndeski or even Lovelock theories. An interesting application for Lovelock theories may be to construct initial data for a binary black hole system and investigate how the weakly coupled Lovelock terms affect the black string instability \cite{Gregory:1993vy}. Another straightforward application could be to extend the methods to the case of systems that involve other compact objects, such as binaries involving scalarized neutron stars in $4\partial$ST theories \cite{Silva:2017uqg}.

In section \ref{sec:xcts_discussion}, we mentioned that the $4\partial$ST version of the extended conformal thin sandwich system suffers from a likely failure of uniqueness of solutions associated with the elliptic equation obtained for the initial value of the scalar field. It would be interesting to apply the machinery of bifurcation theory to this system to better understand the nature of this non-uniqueness.

Nevertheless, the results presented in this paper for the CTT and CTS systems have the obvious possible application to construct black hole binary initial data numerically that could be evolved using the modified harmonic formulation of \cite{Kovacs:2020ywu}. The puncture data (which is based on the CTT approach) has the advantage that it only requires solving a single elliptic equation (the Hamiltonian constraint) coupled to a set of algebraic equations. The main setback of this method (besides the presence of junk radiation in the data) that there appears to be no simple and natural candidate for the scalar field initial data. In the conformal thin sandwich method, one needs to solve a coupled system of elliptic equations for the conformal factor and the shift vector which may be computationally more costly than the puncture method. However, if the $4$-derivative EFT couplings are sufficiently weak, these equations may be solved iteratively, starting from an initial data in General Relativity. Furthermore, the free part of the data may be fixed using standard methods of General Relativity, e.g. by "superposing" stationary scalar hairy black hole solutions. 

\subsection*{Acknowledgments} 

I would like to thank Harvey Reall for helpful conversations, encouragement throughout this work and for his comments on a draft. I am also grateful to Justin Ripley for useful feedback on the manuscript. The author acknowledges financial support from the George and Lilian Schiff Studentship.

\begin{appendices}

\section{Evolution equations of scalar-tensor effective theories in a $3+1$ form}\label{app:adm}

In this appendix we provide the evolution equations of the $4\partial$ST theories written in terms of the ADM variables. These equations can be found in e.g. \cite{Berti2020} for the theories with $f_1(\phi)\equiv 0$. The main reason for giving these equations is that our convention for the scalar field is slightly different from that of \cite{Berti2020}.

Following \cite{Berti2020}, it is useful to introduce
\bea
{\cal A}_{kl}&=&\frac{1}{\alpha}(\partial_t-\mathcal{L}_\beta)K_{kl}-K_{km}K^m{}_l-\frac{1}{\alpha}D_kD_l\alpha \\
{\cal A}_\phi&=&\frac{1}{\alpha}(\partial_t-\mathcal{L}_\beta)K_\phi+\frac{1}{2\alpha}D^k\alpha D_k\phi
\eea
The evolution equations of $4\partial$ST theories are the spatial components of the gravitational equation of motion and the scalar equation of motion. These can be written as
\bea
E_{ij}&\equiv & \frac{\partial(\gamma^{-1/2}\pi_{ij})}{\partial K_{kl}}{\cal A}_{kl}+\frac{\partial(\gamma^{-1/2}\pi_{ij})}{\partial K_\phi}{\cal A}_{\phi}-{\cal F}_{ij}=0 \\
E_\phi & \equiv &\frac{\partial(\gamma^{-1/2}\pi_{\phi})}{\partial K_{kl}}{\cal A}_{kl}+\frac{\partial(\gamma^{-1/2}\pi_{\phi})}{\partial K_{\phi}}{\cal A}_{\phi}-{\cal F}_\phi=0
\eea
with
\bea
{\cal F}^i{}_{j}&=&-\frac12(1+2\epsilon_1 f_1(\phi) X)D^i\phi D_j\phi+\frac12 V(\phi)\gamma^i_j+\frac14\gamma^i_j\left(D^k\phi D_k\phi-4K_\phi^2\right)\left(1+\epsilon_1 f_1(\phi) X\right) \nonumber \\
& &-\frac14\gamma^{ii_1i_2}_{jj_1j_2}\biggl[\left(R[D]_{i_1i_2}{}^{j_1j_2}+2 K_{i_1}{}^{j_1}K_{i_2}{}^{j_2}\right)\left(1+16 \epsilon_2 f_2^{\prime\prime}(\phi)K_\phi^2\right) \nonumber \\
& & +16\epsilon_2\left(\gamma_{i_1}^{j_3}\gamma_{i_3}^{j_1}+\gamma^{j_1j_3}\gamma_{i_1i_3}\right)\left(-K_{j_3}{}^k D_k f_2(\phi)+2D_{j_3}( f_2^\prime(\phi)K_\phi )\right)D^{i_3}K_{i_2}{}^{j_2}\biggr] \\
{\cal F}_{\phi}&=&-V^\prime(\phi)+(1+6\epsilon_1 f_1(\phi)X-8\epsilon_1 f_1(\phi)K_\phi^2)\left(D^kD_k\phi-2KK_\phi\right)-3\epsilon_1 f_1^\prime(\phi)X^2 \nonumber \\
& &-8\epsilon_1 f_1(\phi)K_\phi\left(K_i{}^jD_j\phi D^i\phi-2D_iK_\phi D^i\phi\right)+2\epsilon_1 f_1(\phi)\gamma^{i_1i_2}_{j_1j_2}D_{i_1}\phi D^{j_1}\phi\left(D_{i_2}D^{j_2}\phi-2K_\phi K_{i_2}{}^{j_2}\right) \nonumber \\ 
& &+8\epsilon_2 f_2^\prime(\phi)\gamma^{i_1i_2i_3}_{j_1j_2j_3}D^{j_1}K_{i_2}{}^{j_2}D_{i_1} K_{i_3}{}^{j_3}
\eea
and the coefficients of the time derivatives are
\bea
\frac{\partial(\gamma^{-1/2}\pi^i{}_{j})}{\partial K^k{}_{l}}&=&\gamma^{il}_{jk}-4\epsilon_2\gamma^{ili_1}_{jkj_1}\left(D_{i_1}D^{j_1} f_2(\phi)-2 f_2^\prime(\phi)K_\phi K_{i_1}{}^{j_1}\right) \\
\frac{\partial(\gamma^{-1/2}\pi^i{}_{j})}{\partial K_\phi}=\frac{\partial(\gamma^{-1/2}\pi_\phi)}{\partial K_i{}^j}&=&-2\epsilon_2 f_2^\prime(\phi)\gamma^{ii_1i_2}_{jj_1j_2}\left(R[D]_{i_1i_2}{}^{j_1j_2}+2K_{i_1}{}^{j_1}K_{i_2}{}^{j_2}\right) \\
\frac{\partial(\gamma^{-1/2}\pi_\phi)}{\partial K_\phi}&=&-2\left(1+6\epsilon_1 f_1(\phi)X+2 \epsilon_1 f_1(\phi)(D\phi)^2\right)
\eea

\section{Conformal decomposition}\label{app:conf}

In this section, we collected some identities for converting the elliptic equations to their conformal form. First of all, we have
\be
X\equiv -\frac12 (\nabla\phi)^2=2\psi^{-12}{\tilde K_\phi}^2-\frac12\psi^{-4}\tilde \gamma^{ij}\tilde D_i\phi \tilde D_j\phi
\ee
For a general scalar function $\Phi$, we have
\bea
D_iD^j\Phi&=\psi^{-4}\tilde D_i \tilde D^j \Phi-2\psi^{-4} \tilde D_k \Phi\left(\delta_i^k \tilde D^j \ln \psi+\tilde \gamma^{jk}\tilde D_i \ln\psi-\delta_i^j\tilde \gamma^{kl}\tilde D_l \ln\psi\right)
\eea
and in particular,
\be
D_iD^i\Phi=\psi^{-4}\left(\tilde \gamma^{ij}\tilde D_i \tilde D_j \Phi+2\tilde \gamma^{ij}\tilde D_i\ln\psi \tilde D_j \Phi\right)
\ee
The conformal transformation rule for the Riemann tensor is
\bea
R[D]_{i_1i_2}{}^{j_1j_2}&=&\psi^{-4}R[\tilde D]_{i_1i_2}{}^{j_1j_2}-8\psi^{-5}\delta_{[i_1}^{[j_1}\tilde D_{i_2]}\tilde D^{j_2]}\psi \nonumber \\
& &+24\psi^{-6}\delta_{[i_1}^{[j_1}\tilde D_{i_2]}\psi\tilde D^{j_2]}\psi-4\delta_{i_1i_2}^{j_1j_2}\psi^{-6}\tilde\gamma^{kl}\tilde D_k\psi\tilde D_l\psi
\eea
To convert products of the extrinsic curvature to conformal variables, the following may be useful:
\bea
\delta^{i_1i_2i_3}_{j_1j_2j_3}K_{i_1}{}^{j_1}K_{i_2}{}^{j_2}&=&\frac29 K^2\delta^{i_3}_{j_3}-\frac23 \psi^{-6} K \tilde A^{i_3}_{j_3}+2\psi^{-12}\tilde A^{i_3}{}_k\tilde A^k{}_{j_3}-\psi^{-12}\tilde A^{k}{}_l\tilde A^l{}_{k} \delta^{i_3}_{j_3}
\eea
Finally, for the derivatives of the extrinsic curvature, one can use
\bea
D^jK_{i_1}{}^{j_1}&=&\psi^{-10}\tilde D^j \tilde A_{i_1}{}^{j_1}+\frac13 \psi^{-4}\tilde \gamma_{i_1}{}^{j_1} \tilde D^j K-6\psi^{-10}\tilde A_{i_1}{}^{j_1}\tilde D^j\ln\psi \nonumber \\
& &+2\psi^{-10}\left(\tilde\gamma^{jj_1}\tilde A_{i_1}{}^l\tilde D_l\ln\psi+\tilde \gamma_{i_1}{}^j\tilde A^{j_1l}\tilde D_l\ln\psi-\tilde A_{i_1}{}^j\tilde D^{j_1}\ln\psi-\tilde A^{jj_1}\tilde D_{i_1}\ln\psi\right)
\eea

\end{appendices}

\bibliographystyle{JHEP}
\bibliography{refs}
\end{document}